\documentclass[11pt]{article}%
\usepackage{amsmath}
\usepackage{amsfonts}
\usepackage{amssymb}
\usepackage{graphicx}%
\providecommand{\U}[1]{\protect\rule{.1in}{.1in}}
\newtheorem{theorem}{Theorem}[section]

\newtheorem{corollary}[theorem]{Corollary}

\newtheorem{definition}[theorem]{Definition}
\newtheorem{example}[theorem]{Example}

\newtheorem{lemma}[theorem]{Lemma}

\newtheorem{proposition}[theorem]{Proposition}
\newtheorem{remark}[theorem]{Remark}

\numberwithin{equation}{section}
\newenvironment{proof}[1][Proof]{\noindent\textbf{#1.} }{\ \rule{0.5em}{0.5em}}
\begin{document}

\title{Essential selfadjointness of the graph-Laplacian}
\author{Palle E. T. Jorgensen\\Department of Mathematics\\University of Iowa\\Iowa City, IA 52242-1419 USA}
\date{}
\maketitle

\begin{abstract}
\footnote{{\small The author was partially supported by a grant from the US
National Science Foundation.}
\par
\quad{\small Math Subject Classification (2000): 47C10, 47L60, 47S50, 60J20,
81Q15, 81T75, 82B44, 90B15}
\par
\quad{\small Keywords: Graphs, conductance, network of resistors, resistance
metric, discrete dynamical systems, Hermitian operator, unbounded operators,
Hilbert space, deficiency indices, infinite Heisenberg matrices.}}We study the
operator theory associated with such infinite graphs $G$ as occur in
electrical networks, in fractals, in statistical mechanics, and even in
internet search engines. Our emphasis is on the determination of spectral data
for a natural Laplace operator associated with the graph in question. This
operator $\Delta$ will depend not only on $G$, but also on a prescribed
positive real valued function $c$ defined on the edges in $G$. In electrical
network models, this function $c$ will determine a conductance number for each
edge. We show that the corresponding Laplace operator $\Delta$ is
automatically essential selfadjoint. By this we mean that $\Delta$ is defined
on the dense subspace $\mathcal{D}$ (of all the real valued functions on the
set of vertices $G^{0}$ with finite support) in the Hilbert space $l^{2}%
(G^{0})$. The conclusion is that the closure of the operator $\Delta$ is
selfadjoint in $l^{2}(G^{0})$, and so in particular that it has a unique
spectral resolution, determined by a projection valued measure on the Borel
subsets of the infinite half-line. We prove that generically our graph Laplace
operator $\Delta=\Delta_{c}$ will have continuous spectrum. For a given
infinite graph $G$ with conductance function $c$, we set up a system of finite
graphs with periodic boundary conditions such the finite spectra, for an
ascending family of finite graphs, will have the Laplace operator for $G$ as
its limit. \newline\vspace{4in}

\end{abstract}
\tableofcontents

\newpage

\newpage

\section{Introduction\label{Intro}}

The infinite graphs we consider live on a fixed countable infinite set, say
$L$. Starting with such a set $L$ (subject to certain axioms, listed below),
we get a notion of edges as follows: Select distinguished pairs of points in
$L$, say $x$ and $y$, and connect them by a \textquotedblleft
line,\textquotedblright\ called edge. In physics, when a vertex $x$ is given,
the set of vertices connected to $x$ with one \textquotedblleft
edge\textquotedblright\ is called a set of neighbors, or nearest neighbors.
Initially we do not assign direction to the edges. So, as it stands, an edge
$e$ is defined as a special subset $\{x,y\}$ for selected points $x,y$ in $L$.
Think \textquotedblleft nearest\textquotedblright\ neighbors!

A direction is only assigned when we also introduce a function $I$ on edges
$e$, and then this function $I$ is assumed to satisfy $I(x,y)=-I(y,x)$. In
electrical networks, such a function $I$ may represent a current induced by a
potential which is introduced on a graph with fixed resistors. So only if a
current function $I$ is introduced can we define a direction to edges, as
follows: We specify source $s(e)=x,$ and terminal vertex $t(e)=y$ if
$I(x,y)>0$. meaning that the current flows from $x$ to $y$.

In this paper we study the operator theory of infinite graphs $G$, with
special emphasis on a natural Laplace operator associated with the graph in
question. This operator will depend not only on $G$, but also on a positive
real valued function $c$ defined on the edges in $G$. In electrical network
models, the function $c$ will determine a conductance number for each edge
$e$. If $e=(xy)$ connects vertices $x$ and $y$ in $G$, the number $c(e)$ is
the reciprocal of the resistance between $x$ and $y$. Hence prescribing a
conductance leads to classes of admissible flows in $G$. When they are
determined from Ohm's law, and the Kirchhoff laws, it leads to a measure of
energy, and to an associated graph Laplacian. We identify the Hilbert space
$\mathcal{H}(G)$ which offers a useful spectral theory, and our main result is
a theorem to the effect that the graph Laplacian is essentially selfadjoint,
i.e., that its operator closure is a selfadjoint operator in $\mathcal{H}(G)$.

Let $G=(G^{0},G^{1})$ be an infinite graph, $G^{0}$ for vertices, and $G^{1}$
for edges. Every $x$ in $G^{0}$ is connected to a set $\operatorname*{nbh}(x)$
of other vertices by a finite number of edges, but points in
$\operatorname*{nbh}(x)$ are different from $x$; i.e., we assume that $x$
itself is excluded from $\operatorname*{nbh}(x)$; i.e., no $x$ in $G^{0}$ can
be connected to itself with a single edge. Let $c$ be a conductance function
defined on $G^{1}$.

Initially, the graph $G$ will not be directed, but when a conductance is
fixed, and we study induced current flows, then these current flows will give
a direction to the edges in $G$. But the edges in $G$ itself do not come with
an intrinsic direction.

We show that the Laplace operator $\Delta=\Delta_{c}$ is automatically
essentially selfadjoint. By this we mean that $\Delta$ is defined on the dense
subspace $\mathcal{D}$ (of all the real valued functions on $G^{0}$ with
finite support) in the Hilbert space $\mathcal{H}=\mathcal{H}(G)$%
:$=l^{2}(G^{0})$. The explicit formula for the graph Laplacian $\Delta
=\Delta_{\left(  G,c\right)  }$ is given in (\ref{Eq3.6}) in section
\ref{MainTheo} below which also discusses the appropriate Hilbert spaces. The
conclusion is that the closure of the operator $\Delta$ is selfadjoint in
$\mathcal{H}$, and so in particular that it has a unique spectral resolution,
determined by a projection valued measure on the Borel subsets of the infinite
half-line $\mathbb{R}_{+}$; i.e., the spectral measure takes values in the
projections in the Hilbert space: $=l^{2}(G^{0})$. We work out the measure.

In contrast, we note that the corresponding Laplace operator in the continuous
case is not essentially selfadjoint. This can be illustrated for example with
$\Delta=-(d/dx)^{2}$ on the domain $\mathcal{D}$ of consisting of all $C^{2}%
$-functions on the infinite half-line $\mathbb{R}_{+}$ which vanish with their
derivatives at the end points. Then the Hilbert space is $L^{2}(\mathbb{R}%
_{+})$.

So our graph theorem is an instance where the analogy between the continuous
case and the discrete case breaks down.

A second intrinsic issue for the operator theory of infinite graphs $G$, is
that generically our graph Laplace operator $\Delta= \Delta_{c}$ will have
continuous spectrum. We prove this by identifying a covariance system which
implies that the spectrum of the corresponding Laplace operator will in fact
be absolute continuous with respect to Lebesgue measure on the half-line.

In a third theorem, for a given infinite graph $G$ with conductance function
$c$, we set up a system of finite graphs with periodic boundary conditions
such the finite spectra, for an ascending family of finite graphs, will have
the Laplace operator for $G$ as its limit.

\section{Assumptions\label{Assumptions}}

In order to do computations and potential theory on infinite graphs $G$, it
has been useful to generalize the continuous Laplacian $\Delta$ from
Riemannian geometry \cite{AnCo04} to a discrete setting \cite{BHS05},
\cite{CuSt07}, \cite{Kig03}, \cite{HKK02}. However the infinities for graphs
suggest an analogy to non-compact Riemannian manifolds, or manifolds with boundary.

Once the graph Laplacian is made precise as a selfadjoint operator it makes
sense to ask for exact formulas for the spectrum of $\Delta$. Our Laplace
operator $\Delta=\Delta_{c}$ is associated with a fixed system $(G,c)$ where
vertices and edges are specified as usual, $G=(G^{\left(  0\right)
},G^{\left(  1\right)  })$; and with a fixed conductance function
$c:G^{\left(  1\right)  }\rightarrow\mathbb{R}_{+}$. See (\ref{Eq3.6}) below
for a formula.

And as usual our Laplace operator, $\Delta=\Delta_{c}$ is densely defined in
the Hilbert space $\ell^{2}(G^{\left(  0\right)  })$ of all square-summable
sequences on the vertices of $G$; and if $G$ is infinite, $\Delta_{c}$ is not
defined everywhere in $\ell^{2}$, but rather it has a dense domain
$\mathcal{D}$ in $\ell^{2}$. We show in the next section that $\Delta_{c}$ is
essentially selfadjoint for all choices of conductance function $c$.

By a graph $G$ we mean a set $G^{\left(  0\right)  }$ of vertices, and a set
$G^{\left(  1\right)  }$ of edges. Edges $e$ consist of pairs $s,y\in
G^{\left(  0\right)  }$. We write $e=\left(  xy\right)  $; and if $\left(
xy\right)  \in G^{\left(  1\right)  }$ we say that $x\sim y$.

Assumptions

(i) $x\not \sim x$ (i.e.; $\left(  xx\right)  \not \in G^{\left(  1\right)  }$).

(ii) For every $x\in G^{\left(  0\right)  }$, $\left\{  y\in G^{\left(
0\right)  }|y\sim x\right\}  $ is finite.

(iii) Points $x,y\in G^{\left(  0\right)  }$ for which there is a finite path
$x_{0},x_{1},x_{2},\ldots,x_{n}$ with $x_{0}=x,~x_{n}=y$, and $\left(
x_{i}x_{i+1}\right)  \in G^{\left(  1\right)  },~i=0,\ldots n-1$, are said to
be \textit{connected.}

(iv) We will assume that all connected components in $G^{\left(  0\right)  }$
are \textit{infinite}; or else that $G^{\left(  0\right)  }$ is already connected.

\section{The Main Theorem\label{MainTheo}}

\subsection{\textbf{The Graph Laplacian\label{GraphLapl}}}

In this section we specify a fixed graph $G$ (infinite in the non-trivial
case) and an associated conductance function $c$. The associated graph
Laplacian $\Delta_{c}$ will typically be an unbounded Hermitian operator with
dense domain.

Our assumptions will be as above, and when the Hilbert spaces have been
selected, our main theorem states that the graph Laplacian $\Delta_{c}$ is
essentially selfadjoint; i.e., the operator closure, also denoted $\Delta_{c}%
$, is a selfadjoint operator. In sections 5--8 we obtain consequences and applications.

The interpretation of this results in terms of boundary conditions will be
given in section \ref{BoundaryCond} below. It means that $\Delta_{c}$ has a
well defined and unique (up to equivalence) spectral resolution. Then the next
objective is to find \textit{the spectrum} of the operator $\Delta_{c}$. And a
method for finding \textit{spectrum} is based on \textquotedblleft
covariance.\textquotedblright\ Covariance is used on other spectral problems
in mathematical physics, and it offers useful ways of getting global formulas
for spectrum. As we will see, infinite models typically have graph Laplacians
with continuous spectrum.

In the finite case, of course the spectrum is the set of roots in a
characteristic polynomial, but unless there is some group action, it is
difficult to solve for roots by \textquotedblleft bare
hands;\textquotedblright\ and even if we do, only the occurrence of groups
offers insight.

A second approach to the finding spectra of graph Laplacians is
\textquotedblleft renormalization:\textquotedblright\ Renormalization of
hierarchical systems of electrical networks comes into play each time one
passes to a new scale (upwards or downwards). This requires additional
structure, such as is found in iterated function systems (IFSs), (see
\cite{BHS05}, \cite{DuJo07a}, \cite{JoPe98}, \cite{Kig03}), i.e., specified
finite systems of affine transformations in Euclidean space that are then
iterated recursively.

When the mappings are so iterated on a given graph, the iterations may then be
interpreted as scales in an infinite graph: (post-)composition of similarity
mappings takes us further down the branches of a tree like structure in path
space. We get martingale constructions as instances of renormalization.

\begin{theorem}
\label{Theo3.1}The graph Laplacian $\Delta=\Delta_{\left(  G,c\right)  }$ is
essentially selfadjoint.
\end{theorem}

\begin{proof}
To get started we recall the setting. Given:

\begin{itemize}
\item[ ] $G$: a fixed infinite graph. (It may be finite, but in this case the
conclusion follows from finite-dimensional linear algebra.)

\item[ ] $G=\left(  G^{\left(  0\right)  },G^{\left(  1\right)  }\right)  $.

\item[ ] $G^{\left(  0\right)  }$: the set of vertices in $G$.

\item[ ] $G^{\left(  1\right)  }$: the set of edges in $G$.
\end{itemize}

If $x,y\in G^{\left(  0\right)  }$ is a given pair, we say that $x\sim y$ when
$e=\left(  xy\right)  \in G^{\left(  1\right)  }$.

For $x\in G^{\left(  0\right)  }$, set
\begin{equation}
\operatorname*{nbh}\left(  x\right)  =\left\{  y\in G^{\left(  0\right)
}|y\sim x\right\}  \text{.} \label{Eq3.1}%
\end{equation}

Our standing assumptions are as follows:\medskip

(a) $\operatorname*{nbh}\left(  x\right)  $ is finite.

(b) $x\not \in \operatorname*{nbh}\left(  x\right)  $.

\begin{itemize}
\item[ ] $\mathcal{H}=\ell^{2}\left(  G^{\left(  0\right)  }\right)  =$ all
functions $v:G^{\left(  0\right)  }\rightarrow\mathbb{C}$ such that%
\begin{equation}
\sum_{x\in G^{\left(  0\right)  }}\left\vert v\left(  x\right)  \right\vert
^{2}<\infty. \label{Eq3.2}%
\end{equation}

\end{itemize}

Set
\begin{equation}
\langle u,v\rangle\text{:}=\sum_{x\in G^{\left(  0\right)  }}\overline
{u\left(  x\right)  }v\left(  x\right)  ,~\forall u,v\in\ell^{2}\left(
G^{\left(  0\right)  }\right)  \text{.} \label{Eq3.3}%
\end{equation}
By $\mathcal{H}$ we refer to the completed Hilbert space $\ell^{2}\left(
G^{\left(  0\right)  }\right)  $.

\begin{itemize}
\item[ ] $\mathcal{D}$:$=$ the set of all finitely supported $v\in\mathcal{H}%
$; i.e., $v$ is in $\mathcal{D}$ iff $\exists F\subset G^{\left(  0\right)  }%
$, $F=F_{v}$ some \textit{finite} subset such that
\[
v\left(  x\right)  =0,~\forall x\in G^{\left(  0\right)  }\backslash F\text{.}%
\]

\item[ ] $e_{x}$:$=\delta_{x}=$ Dirac mass, defined by
\begin{equation}
e_{x}\left(  y\right)  =\left\{
\begin{array}
[c]{l}%
1\text{ if }y=x\\
0\text{ if }y\not =x\text{.}%
\end{array}
\right.  \label{Eq3.4}%
\end{equation}

\item[ ] $c:G^{\left(  1\right)  }\rightarrow\mathbb{R}_{+}$ is a fixed
function taking positive values. In network models, the function $c$ is
\textit{conductance}; i.e., the reciprocal of resistance.
\end{itemize}

Assumption (symmetry): $c\left(  xy\right)  =c\left(  yx\right)
,~\forall\left(  xy\right)  \in G^{\left(  1\right)  }$.
\begin{equation}
\Delta=\Delta_{\left(  G,c\right)  } \label{Eq3.5}%
\end{equation}
is the Laplacian, and is defined on $\mathcal{D}$ as follows:
\begin{equation}
\left(  \Delta v\right)  \left(  x\right)  \text{:}=\sum_{y\sim x}c\left(
xy\right)  \left(  v\left(  x\right)  -v\left(  y\right)  \right)  ,~\forall
v\in\mathcal{D},~\forall x\in G^{\left(  0\right)  }\text{.} \label{Eq3.6}%
\end{equation}
\hfill\ 
\end{proof}

\subsection{\textbf{Lemmas\label{Lemmas}}}

We will need some lemmas:

\begin{lemma}
\label{Lem3.2}The operator $\Delta$ is Hermitian symmetric on $\mathcal{D}$,
and it is positive semidefinite. Specifically, the following two properties
hold$:$%
\begin{equation}
\langle\Delta u,v\rangle_{\ell^{2}}=\langle u,\Delta v\rangle_{\ell^{2}%
},~\forall u,v\in\mathcal{D}\text{;} \label{Eq3.7}%
\end{equation}
and
\begin{equation}
\langle u,\Delta u\rangle_{\ell^{2}}\geq0,~\forall u\in\mathcal{D}\text{.}
\label{Eq3.8}%
\end{equation}

\end{lemma}

\begin{proof}
Both assertions are computations:

In (\ref{Eq3.7}),%
\begin{align*}
\langle\Delta u,v\rangle_{\ell^{2}}  &  ={{\underset{%
\genfrac{}{}{0pt}{}{x,y\in G^{\left(  0\right)  }}{x\sim y}%
}{\sum\sum}c\left(  x,y\right)  \left(  \overline{u\left(  x\right)
}-\overline{u\left(  y\right)  }\right)  v\left(  x\right)  }}\\
&  =\underset{x\sim y}{\sum\sum}c\left(  xy\right)  \overline{u\left(
x\right)  }v\left(  x\right)  -\underset{x\sim y}{\sum\sum}\overline{u\left(
y\right)  }c\left(  xy\right)  v\left(  x\right) \\
&  =\underset{x\sim y}{\sum\sum}\overline{u\left(  x\right)  }c\left(
xy\right)  v\left(  x\right)  -\underset{x\sim y}{\sum\sum}\overline{u\left(
x\right)  }c\left(  xy\right)  v\left(  y\right) \\
&  =\sum_{x\in G^{\left(  0\right)  }}\overline{u\left(  x\right)  }%
\sum_{y\sim x}c\left(  xy\right)  \left(  v\left(  x\right)  -v\left(
y\right)  \right) \\
&  =\langle u,\Delta v\rangle_{\ell^{2}}\text{.}%
\end{align*}

Note that the summation may be exchanged since, for each $x\in G^{\left(
0\right)  }$, the set of neighbors $\operatorname*{nbh}\left(  x\right)  $ is finite.

In (\ref{Eq3.8}),%
\begin{align*}
&  \langle u,\Delta u\rangle_{\ell^{2}}\\
&  =\underset{x\sim y}{\sum\sum}\overline{u\left(  x\right)  }c\left(
xy\right)  \left(  u\left(  x\right)  -u\left(  y\right)  \right) \\
&  =\sum_{x\in G^{\left(  0\right)  }}B_{c}\left(  x\right)  \left\vert
u\left(  x\right)  \right\vert ^{2}-\underset{x\sim y}{\sum\sum}%
\overline{u\left(  x\right)  }c\left(  xy\right)  u\left(  y\right)  \text{,}%
\end{align*}
where
\begin{equation}
B_{c}\left(  x\right)  =\sum_{y\sim x}c\left(  xy\right)  ,~x\in G^{\left(
0\right)  }\text{.} \label{Eq3.9}%
\end{equation}

The second term in the computation may be estimated with the use of
Cauchy-Schwarz as follows: Setting
\begin{equation}
\mathcal{E}_{c}\left(  u\right)  \text{:}=\underset{%
\genfrac{}{}{0pt}{}{\text{all~}x,y}{\text{s.t. }x\sim y}%
}{\sum\sum}c\left(  xy\right)  \left\vert u\left(  x\right)  -u\left(
y\right)  \right\vert ^{2}\text{;} \label{Eq3.10}%
\end{equation}
we show that
\begin{equation}
2\langle u,\Delta u\rangle=\mathcal{E}_{c}\left(  u\right)  \geq0\text{.}
\label{Eq3.11}%
\end{equation}

Indeed using the conditions on $c$:$~G^{\left(  1\right)  }\rightarrow
\mathbb{R}+$

\begin{itemize}
\item $c\left(  xy\right)  =c\left(  yx\right)  ,$ $\forall\left(  xy\right)
\in G^{\left(  1\right)  }$;

\item $c\left(  xx\right)  =0,$ $\forall x\in G^{\left(  0\right)  }$;

\item $c\left(  xy\right)  >0,$ $\forall\left(  xy\right)  \in G^{\left(
1\right)  }$,
\end{itemize}

we get
\begin{align*}
2\langle u,\Delta u\rangle &  =2\sum_{x\in G^{\left(  0\right)  }}B_{c}\left(
x\right)  \left\vert u\left(  x\right)  \right\vert ^{2}-2\underset{x\sim
y}{\sum\sum}\overline{u\left(  x\right)  }c\left(  xy\right)  u\left(
y\right) \\
&  =2\sum_{x\in G^{\left(  0\right)  }}B_{c}\left(  x\right)  \left\vert
u\left(  x\right)  \right\vert ^{2}-2\operatorname{Re}\underset{x\sim y}%
{\sum\sum}\overline{u\left(  x\right)  }c\left(  xy\right)  u\left(  y\right)
\\
&  =\underset{x\sim y}{\sum\sum}c\left(  x,y\right)  \left(  \left\vert
u\left(  x\right)  \right\vert ^{2}-\overline{u\left(  x\right)  }u\left(
y\right)  -\overline{u\left(  y\right)  }u\left(  x\right)  +\left\vert
u\left(  y\right)  \right\vert ^{2}\right) \\
&  =\underset{xy}{\sum\sum}c\left(  xy\right)  \left\vert u\left(  x\right)
-u\left(  y\right)  \right\vert ^{2}=\mathcal{E}_{c}\left(  u\right)  \text{.}%
\end{align*}

\end{proof}

For the general theory of unbounded Hermitian operators and their extensions,
we refer the reader to \cite{Jor78}, \cite{Nel69}, \cite{Sto51}.

\begin{definition}
\label{Def3.3} If $\Delta$ is an operator with dense domain $\mathcal{D}$ in a
Hilbert space $\mathcal{H}$, we define its adjoint operator $\Delta^{\ast}$
by$:$

A vector $v$ is in the domain $\operatorname*{dom}\left(  \Delta^{\ast
}\right)  $ iff there is a constant $K$ such that
\begin{equation}
\left\vert \langle v,\Delta u\rangle\right\vert \leq K\left\Vert u\right\Vert
,~\forall u\in\mathcal{D}\text{.} \label{Eq3.12}%
\end{equation}
When \emph{(}\ref{Eq3.12}\emph{)}\thinspace holds, then by Riesz, there is a
unique $w:=\Delta^{\ast}v$ such that
\begin{equation}
\langle w,u\rangle=\langle v,\Delta u\rangle,~\forall u\in\mathcal{D}\text{.}
\label{Eq3.13}%
\end{equation}
Note that since $\mathcal{D}$ is dense in $\mathcal{H}$, $w\left(
=:\Delta^{\ast}v\right)  $ is uniquely determined by \emph{(}\ref{Eq3.12}%
\emph{)}.
\end{definition}

\begin{lemma}
\label{Lem3.4}In the case of $\Delta=\Delta_{\left(  G,c\right)  }$ and
$\mathcal{H}=\ell^{2}\left(  G^{\left(  0\right)  }\right)  $, the vector
$\Delta^{\ast}v$ for $v\in\operatorname*{dom}\left(  \Delta^{\ast}\right)  $
is given by the expression
\begin{equation}
\left(  \Delta^{\ast}v\right)  \left(  x\right)  =\sum_{y\sim x}c\left(
xy\right)  \left(  v\left(  x\right)  -v\left(  y\right)  \right)  \text{.}
\label{Eq3.14}%
\end{equation}

\end{lemma}

\begin{proof}
Since the sum in (\ref{Eq3.13}) is finite, the RHS is well defined if
$v\in\operatorname*{dom}\left(  \Delta^{\ast}\right)  $. Since $\Delta^{\ast
}v\in\mathcal{H}$,
\begin{equation}
\sum_{x\in G^{\left(  c\right)  }}\left\vert \left(  \Delta^{\ast}v\right)
\left(  x\right)  \right\vert ^{2}<\infty\text{.} \label{Eq3.15}%
\end{equation}

Set $w\left(  x\right)  $:$=\sum_{y\sim x}c\left(  xy\right)  \left(  v\left(
x\right)  -v\left(  y\right)  \right)  $.

We claim that (\ref{Eq3.12}) then holds. Indeed
\begin{align*}
\langle w,u\rangle_{\ell^{2}}  &  =\sum_{x\in G^{\left(  0\right)  }}\left(
\sum_{y\sim x}c\left(  xy\right)  \left(  \overline{v\left(  x\right)
}-\overline{v\left(  y\right)  }\right)  \right)  u\left(  x\right) \\
&  =\sum_{x}\overline{v\left(  x\right)  }\sum_{y\sim x}c\left(  xy\right)
\left(  u\left(  x\right)  -u\left(  y\right)  \right) \\
&  =\langle v,\Delta u\rangle_{\ell^{2}}%
\end{align*}
(by the exchange of summation and Lemma \ref{Lem3.2}).
\end{proof}

\begin{lemma}
\label{Lem3.5}Let $\left(  G,c\right)  $ and $\Delta=\Delta_{\left(
G,c\right)  }$ here as in the previous lemma. Then the equation
\begin{equation}
\Delta v=-v \label{Eq3.16}%
\end{equation}
does not have non-zero solutions $v\in\ell^{2}\left(  G^{\left(  0\right)
}\right)  $.
\end{lemma}

\begin{proof}
It is immediate from (\ref{Eq3.7}) in Lemma \ref{Lem3.2} that eq.
(\ref{Eq3.16}) does not have non-zero solutions in $\mathcal{D}$, but the
assertion is that there are no non-zero solutions in any bigger subspace.

Also note that every solution in $\ell^{2}\left(  G^{\left(  0\right)
}\right)  $ to eq. (\ref{Eq3.16}) must be in $\operatorname*{dom}\left(
\Delta^{\ast}\right)  $, i.e., the domain of the adjoint of $\Delta$ with
$\mathcal{D}$ as domain.

If $v:G^{\left(  0\right)  }\rightarrow\mathbb{C}$ is a solution to
(\ref{Eq3.16}), then
\begin{equation}
\overline{v\left(  x\right)  }\Delta v\left(  x\right)  =-\left\vert v\left(
x\right)  \right\vert ^{2},~\forall x\in G^{\left(  0\right)  } \label{Eq3.17}%
\end{equation}
which yields $\overline{v\left(  x\right)  }\Delta v\left(  x\right)
\leq0,~\forall x\in G^{\left(  0\right)  }$. Hence $\mathcal{E}_{c}\left(
v\right)  \leq0$; see (\ref{Eq3.10})-(\ref{Eq3.11}). But by (\ref{Eq3.11}),
then $\mathcal{E}_{c}\left(  v\right)  =0$.

It follows from (\ref{Eq3.10}) that $v$ must be constant on every connected
component in $G^{\left(  0\right)  }$. Since all the connected components are
infinite, $v$ must be zero.
\end{proof}

\begin{remark}
\label{Rem3.5.1}We stress that \emph{(}\ref{Eq3.16}\emph{)} may have non-zero
solutions not in $\ell^{2}$. For these solutions $v$, the energy will be unbounded.
\end{remark}

\begin{example}
\label{Examp3.5.2}Let a graph system $\left(  G,c\right)  $ be determined as follows:

\begin{itemize}
\item[ ] $G^{\left(  0\right)  }=\mathbb{N}_{0}=\left\{  0,1,2,\ldots\right\}
$,

\item[ ] $G^{\left(  1\right)  }:\operatorname*{nbh}\left(  0\right)
=\left\{  1\right\}  ,$

\item[ ] \qquad$\;\operatorname*{nbh}\left(  n\right)  =\left\{
n-1,n+1\right\}  $ if $n>0$, and%
\[
c\left(  n,n+1\right)  =n+1\text{.}%
\]
Then the Laplace operator $\Delta_{c}$ will be unbounded in $\ell^{2}$ as
follows from
\[
\Delta_{c}=\left(
\begin{array}
[c]{rrrrrll}%
1 & -1 & 0 & 0 & 0 & \cdots & \\
-1 & 3 & -2 & 0 & 0 & \cdots & \\
0 & -2 & 5 & -3 & 0 & \cdots & \\
0 & 0 & -3 & 7 & -4 & \cdots & \\
\vdots & \vdots & \vdots & \ddots & \ddots & \ddots & \\
&  &  &  & -n & 2n+1 & -\left(  n+1\right)  \\
&  &  &  &  &  & \ddots
\end{array}
\right)
\]
Then
\begin{align*}
\left(  \Delta u\right)  _{0} &  =u_{0}-u_{1}\text{, and}\\
\left(  \Delta u\right)  _{n} &  =\left(  2n+1\right)  u_{n}-nu_{n-1}-\left(
n+1\right)  u_{n+1},~\forall n\geq1\text{.}%
\end{align*}

\end{itemize}
\end{example}

For solving (\ref{Eq3.16}), initialize $v_{0}=1$. Then
\begin{align*}
v_{1} &  =2v_{0}=2,\\
v_{2} &  =\frac{7}{2}\text{, and inductively}\\
v_{n+1} &  =2v_{n}-\left(  \frac{n}{n+1}\right)  v_{n-1}\text{.}%
\end{align*}
We get $v_{1}<v_{2}<\cdots<v_{n-1}<v_{n}<\cdots$ and
\[
v_{n+1}>\left(  2-\frac{n}{n+1}\right)  v_{n}\text{.}%
\]
Hence for the truncated summations for $\ell^{2}$ and $\mathcal{E}$ applied to
this solution $v$; we get
\[
\frac{1}{2}\mathcal{E}_{N}\left(  v\right)  =-\sum_{k=0}^{N}v_{k}^{2}<-N
\]
which tends to $-\infty$.

The following lemma is from the general theory of unbounded operators in
Hilbert space, \cite{Nel69}, \cite{Sto51}, \cite{vN31}.

\begin{lemma}
\label{Lem3.6}Let $\Delta$ be a linear operator in a Hilbert space
$\mathcal{H}$ and defined in a dense domain $\mathcal{D}$.

Then $\Delta$ is essentially selfadjoint \emph{(}i.e., has selfadjoint closure
$\bar{\Delta}$\emph{)} if the following conditions hold:

\emph{(}i\emph{)} $\langle u,\Delta u\rangle\geq0,\forall u\in\mathcal{D}$

\emph{(}ii\emph{)}\thinspace$\dim\left\{  v\in\operatorname*{dom}\left(
\Delta^{\ast}\right)  |\Delta^{\ast}v=-v\right\}  =0$.
\end{lemma}

\begin{proof}
This is in the literature, e.g. \cite{vN31}. The idea is the following, if (i)
is assumed, then there is a well defined bounded operator
\[
T=\left(  I+\bar{\Delta}\right)  ^{-1}%
\]
precisely when (ii) is satisfied.
\end{proof}

In our analysis of the graph Laplacian $\Delta_{c}$ in (\ref{Eq3.6}) we shall
need one more:

\begin{lemma}
\label{Lem3.7}Let $\Delta_{c}$ be as in \emph{(}\ref{Eq3.6}\emph{)}. Then for
all $v\in\mathcal{D}$,
\begin{equation}
\sum_{x\in G^{\left(  0\right)  }}\left(  \Delta_{c}v\right)  \left(
x\right)  =0\text{.} \label{Eq3.18}%
\end{equation}
In fact, when $v$ is fixed, the number of non-zero terms in \emph{(}%
\ref{Eq3.18}\emph{)} is finite.
\end{lemma}

\begin{proof}
The finiteness claim follows from the assumptions on $\left(  G,c\right)  $ we
listed in section \ref{Assumptions}.

A direct computation yields the result:
\begin{align*}
\sum_{x\in G^{\left(  0\right)  }}\left(  \Delta_{c}v\right)  \left(
x\right)   &  =\sum_{x}\sum_{y\sim x}c\left(  xy\right)  \left(  v\left(
x\right)  -v\left(  y\right)  \right) \\
&  =\sum_{x}v\left(  x\right)  \sum_{y\sim x}c\left(  xy\right)  -\sum
_{y}v\left(  y\right)  \sum_{x\sim y}c\left(  xy\right) \\
&  =\sum_{x}v\left(  x\right)  \mathcal{B}_{c}\left(  x\right)  -\sum
_{y}v\left(  y\right)  \sum_{x\sim y}c\left(  yx\right) \\
&  =\sum_{x}v\left(  x\right)  \mathcal{B}_{c}\left(  x\right)  -\sum
_{y}v\left(  y\right)  \mathcal{B}_{c}\left(  y\right) \\
&  =0\text{.}%
\end{align*}

\end{proof}

\section{Operator Theory\label{OpTheory}}

Once the operator theoretic tools are introduced, we show in section
\ref{EnergyForm} below that class of infinite graph systems $(G,c)$ where $G$
is a graph and $c$ is a conductance function (the pair $\left(  G,c\right)  $
satisfying the usual axioms as before), have the spectrum of the associated
Laplace operator $\Delta_{c}$ continuous. This refers to the $\ell^{2}$ space
of $G^{\left(  0\right)  }$, i.e., the Hilbert space is $\ell^{2}(G^{(0)})$
where as usual $G^{\left(  0\right)  }$ denotes the set of vertices.

It is important that $G^{\left(  0\right)  }$ is infinite. Otherwise of course
the spectrum is just the finite set of zeros of the characteristic polynomial.
See Example \ref{Examp6.6} below.

We give an operator theory/spectral theory analysis, with applications, of a
class of graph Laplacians; and we have been motivated by a pioneering paper
\cite{Pow76} which in an exciting way applies graphs and resistor networks to
a problem in quantum statistical mechanics. In one of our results we establish
the essential selfadjointness of a large class of graph Laplacians on graphs
of infinite networks. (A Hermitian symmetric operator with dense domain in
Hilbert space is said to be \textit{essentially selfadjoint} if its closure is
selfadjoint, i.e., if the deficiency indices are $(0,0)$. See Definition
\ref{Def4.1} below! There are many benefits from having the graph Laplacian
$\Delta$ essentially selfadjoint.

Here is a partial list:

(a) We get the benefit of having the spectral resolution for the selfadjoint
closure, also denoted $\Delta$ for notational simplicity.

(b) We get a spectral representation realization of the operator $\Delta$,
i.e., a unitarily equivalent form of $\Delta$ in which an equivalent operator
$\Delta^{\sim}$ may occur in applications. See e.g., \cite{Arv02},
{\cite{PaSc72}}.

(c) We get a scale of Hilbert spaces, $\mathcal{H}_{s}$ for $s$ in
$\mathbb{R}$, defined from the graph of the operator $\Delta^{s}$ where the
fractional power $\Delta^{s}$ is defined by functional calculus applied to the
selfadjoint realization of $\Delta$. See {\cite{Jor04}}.

(d) Gives us a way of computing scales of resistance metrics on electrical
networks realized on infinite graphs, extending tools available previously
only for finite graphs; see {\cite{BoDu49}}.

(e) The case $s=1/2$ yields an exact representation of the energy Hilbert
space associated with a particular system $(G,c)$ and the corresponding graph
Laplacian $\Delta=\Delta_{(G,c)}$.

(f) Gives us a way of computing fractional Brownian motion on graphs, allowing
an analytic continuation in the parameter $s$, and with $s=1/2$ corresponding
to the standard Brownian motion; see e.g., {\cite{DuJo07a}}, {\cite{Jor06b}}.

In the course of the proofs of our main results, we are making use of tools
from the theory of unbounded operators in Hilbert space: von Neumann's
deficiency indices, operator closure, graphs of operators, operator domains,
operator adjoints; and extensions of Hermitian operators with a dense domain
in a fixed complex Hilbert space. Our favorite references for this material
include: {\cite{AnCo04}}, {\cite{Jor77}}, {\cite{Jor78}}, {\cite{JoPe00}},
{\cite{Nel69}}, {\cite{vN31}}, {\cite{Sto51}}. For analysis on infinite graphs
and on fractals, see e.g., {\cite{BHS05}}, {\cite{CuSt07}}, {\cite{DuJo06a}},
{\cite{HKK02}}, {\cite{Hut81}}, {\cite{JoPe98}}, {\cite{JKS07}},
{\cite{Kig03}}, {\cite{BoDu49}}.

\begin{definition}
\label{Def4.1}Let $\Delta$ be a Hermitian linear operator with dense domain
$\mathcal{D}$ in a complex Hilbert space $\mathcal{H}$. Set
\[
\mathcal{D}_{\pm}:\,=\left\{  v_{\pm}\in\operatorname*{dom}\left(
\Delta^{\ast}\right)  |\Delta^{\ast}v_{\pm}=\pm iv_{\pm}\right\}  \text{,}%
\]
where $i=\sqrt{-1}$. Then the two numbers $n_{\pm}:\,=\dim\mathcal{D}_{\pm}$
are called the \emph{deficiency indices}.
\end{definition}

Von-Neumann's theorem states that the initial operator $\Delta$ is essentially
selfadjoint on $\mathcal{D}$ if and only if $n_{+}=n_{-}=0$. It has
selfadjoint extensions defined on a larger domain in $\mathcal{H}$ if and only
if $n_{+}=n_{-}$.

The following two conditions on a Hermitian operator, (A) and (B),
individually imply equal deficiency indices, i.e., $n_{+}=n_{-}$:

(A) For all $v\in\mathcal{D}$, we have the estimate
\[
\langle v,\Delta v\rangle\geq0\text{,}%
\]
i.e., $\Delta$ is semibounded.

(B) There is an operator $J:\mathcal{H}\rightarrow\mathcal{H}$ satisfying the
following four conditions:

\begin{itemize}
\item[ ] (i) $J\left(  u+\alpha v\right)  =Ju+\bar{\alpha}Jv$, for $\forall
u,v\in\mathcal{H}$, $\alpha\in\mathbb{C}$

\item[ ] (ii) $\langle Ju,Jv\rangle=\langle v,u\rangle~\forall u,v\in
\mathcal{H}$; ($J$ is called a conjugation!)

\item[ ] (iii) $J$ maps the subspace $\mathcal{D}$ into itself, and
\[
J\Delta v=\Delta Jv,~\forall v\in\mathcal{D}\text{.}%
\]

\item[ ] (iv) $J^{2}=id$; $J$ is of period $2$.
\end{itemize}

\begin{remark}
\label{Rem4.2}There are many examples \emph{(}see the Appendix\emph{)} where
either \emph{(}A\emph{)}\thinspace or \emph{(}B\emph{)}\thinspace is satisfied
but where the operator $\Delta$ is not essentially selfadjoint.
\end{remark}

Both conditions (A) and (B) hold for a graph Laplacians $\Delta_{c}$, and
Theorem \ref{Theo3.1} states that $\Delta_{c}$ is essentially selfadjoint.

For Riemannian manifolds with boundary, there is a close analogue of the graph
Laplacian $\Delta_{c}$ above; but it is known (see section \ref{BoundaryCond})
that these continuous variants are typically \textit{not} essentially selfadjoint.

Indeed the obstruction to essential selfadjointness in these cases captures a
physical essence of the metric geometry behind the use of Laplace operators in
Riemannian geometry.

\section{The Energy Form\label{EnergyForm}}

\subsection{\textbf{Operators}\label{Operators}}

In section \ref{MainTheo} we proved essential selfadjointness of the graph
Laplacians $\Delta_{c}$. This refers to sequence space $\ell^{2}$, the Hilbert
space of all square-summable sequences indexed by the points in $G^{\left(
0\right)  }$, and equipped with the usual $\ell^{2}$-inner product.

This means that the axioms for $\Delta_{c}$ are such that boundary conditions
at infinity in $G$ are determined by computations on finite subsets of the
vertices in $G$. (In the Appendix, we will contrast this state of affairs with
related but different boundary conditions from quantum mechanics.) Recall that
$\Delta_{c}$ is generally a densely defined Hermitian and unbounded operator.
So in principle there might be non-trivial obstructions to selfadjointness
(other than simply taking operator closure.) Recall (Definition \ref{Def4.1})
that a given Hermitian operator with dense domain is essentially selfadjoint
if and only if the dimension of the each of the two \textquotedblleft defect
eigenspaces\textquotedblright\ is zero.

So that is why we look at the \textquotedblleft minus 1
eigenspaces\textquotedblright\ for the adjoint operator, $\Delta^{\ast}u=-u.$

For potential theoretic computations we need an additional Hilbert space, the
Energy Hilbert space $\mathcal{H}_{\mathcal{E}}$ (details below.) For example
the voltage potentials associated with a fixed graph Laplacian are typically
not in $\ell^{2}(G^{\left(  0\right)  })$ but rather in an associated Energy
Hilbert space. Our Laplace operator $\Delta$ is formally Hermitian in both the
Hilbert spaces $\ell^{2}$ and $\mathcal{H}_{\mathcal{E}}$(the energy Hilbert
space). We show that the Laplace operator $\Delta$ is essentially selfadjoint
both in $\ell^{2}$ and in $\mathcal{H}_{\mathcal{E}}$. In both cases, we take
for dense domain $\mathcal{D}$ the linear subspace of all finitely supported
functions $G^{\left(  0\right)  }\rightarrow\mathbb{R}$.

Our setting and results in this section are motivated by \cite{Pow76} and
\cite{BoDu49}.

There are several distinctions between the two Hilbert spaces: For example,
the Dirac functions $\{\delta_{x}|x\in G^{\left(  0\right)  }\}$ form an
orthonormal basis (ONB) in $\ell^{2}$, but not in $\mathcal{H}_{\mathcal{E}}$.
The implication of this is that our graph Laplacians have different matrix
representations in the two Hilbert spaces. In speaking of \textquotedblleft
matrix representation\textquotedblright\ for an operator in a Hilbert space,
we will always be referring to a chosen ONB.

We shall need the operator $\Delta$ in both guises. One reason for this is
that for infinite graphs, typically the potential function $v$ solving $\Delta
v=\delta_{x}-\delta_{y}$, for pairs of vertices will not be in $\ell^{2}$, but
nonetheless $v$ will have finite energy, i.e., $\mathcal{E}(v)<\infty$,
meaning that the energy form applied to $v$ is finite. Caution: The sequence
$v$ might not be in the $\ell^{2}$-space. Specifics in Example \ref{Examp5.2} below!

When we study the Laplace operator $\Delta$, our questions concern its
spectrum, and its spectral resolution. The spectrum will be contained in the
half-line $[0,\infty)$, but (as we show in examples) it can be unbounded, and
it can have continuous parts mixed in with discrete parts. In case the
conductance function is \textquotedblleft very unbounded,\textquotedblright%
\ as an operator in $\ell^{2}$, it may be necessary to pass to a proper
operator extension in an enlarged Hilbert space to get a different selfadjoint
realization of $\Delta$.

These operator issues only enter in case $\Delta$ is unbounded. The
unboundedness of the operator $\Delta$ is tied in closely with unboundedness
of the conductance function $c$ which is used. Recall $\Delta=\Delta_{c}$
depends on the choice of $c$. If $c$ is unbounded \textquotedblleft at
infinity\textquotedblright\ (on the set $G^{\left(  1\right)  }$ of edges),
then the resistors tend to zero at distant edges. Intuitively, this means that
\textquotedblleft the current escapes to infinity,\textquotedblright\ and we
make this precise in the language of operator theory. Our general setup here
will be as in \cite{Pow76}.

We will need a second Hilbert space, the energy Hilbert space $\mathcal{H}%
_{\mathcal{E}}$. Here the inner product is the energy quadratic form. Since
the energy form evaluated on a function $v$ is defined in terms of the square
of differences $v(x)-v(y)$, it follows that the elements in $\mathcal{H}%
_{\mathcal{E}}$ are really sequences modulo constants.

Let $G=\left(  G^{\left(  0\right)  },G^{\left(  1\right)  }\right)  $ be a
graph satisfying the axioms in section \ref{MainTheo}, with vertices
$G^{\left(  0\right)  }$ and edges $G^{\left(  1\right)  }.$ Let
\begin{equation}
c:G^{\left(  1\right)  }\rightarrow\mathbb{R}_{+} \label{Eq5.1}%
\end{equation}
be a fixed function (called \textit{conductance}.) If $e=\left(  xy\right)
\in G^{\left(  1\right)  }$, we say that $x\sim y$, and the function $c$ must
satisfy $c\left(  xy\right)  =c\left(  yx\right)  $, symmetry. In particular,
for a pair of vertices $x,y$, $c\left(  xy\right)  $ is only defined if $x\sim
y$, i.e., if $\left(  xy\right)  \in G^{\left(  1\right)  }$. For every $x\in
G^{\left(  0\right)  }$, we assume that
\begin{equation}
\operatorname*{nbh}\left(  x\right)  =\left\{  y\in G^{\left(  0\right)
}|y\sim x\right\}  \label{Eq5.2}%
\end{equation}
is \textit{finite}, and that $x\not \in \operatorname*{nbh}\left(  x\right)  $.

Following eq. (\ref{Eq3.9}), we study functions $v:G^{\left(  0\right)
}\rightarrow\mathbb{C}$ for which
\begin{equation}
\mathcal{E}_{c}\left(  v\right)  =\sum_{%
\genfrac{}{}{0pt}{}{\text{all }x,y}{\text{s.t. }x\sim y}%
}c\left(  xy\right)  \left\vert v\left(  x\right)  -v\left(  y\right)
\right\vert ^{2}<\infty\text{.} \label{Eq5.3}%
\end{equation}

Clearly we must work with functions on $G^{\left(  0\right)  }$ modulo
constants. Setting
\begin{equation}
\mathcal{E}_{c}\left(  u,v\right)  :\,=\sum_{%
\genfrac{}{}{0pt}{}{\text{all }x,y}{x\sim y}%
}c\left(  xy\right)  \left(  \overline{u\left(  x\right)  }-\overline{u\left(
y\right)  }\right)  \left(  v\left(  x\right)  -v\left(  y\right)  \right)
\text{,} \label{Eq5.4}%
\end{equation}
we get
\begin{equation}
\left\vert \mathcal{E}_{c}\left(  u,v\right)  \right\vert ^{2}\leq
\mathcal{E}_{c}\left(  u\right)  \mathcal{E}_{c}\left(  v\right)  ,~\forall
u,v\in\mathcal{D}\text{,} \label{Eq5.5}%
\end{equation}
by Schwarz' inequality.

Setting
\begin{equation}
\langle u,v\rangle_{\mathcal{E}}:~=\mathcal{E}_{c}\left(  u,v\right)
\label{Eq5.6}%
\end{equation}
we get an inner-product, and an associated Hilbert space $\mathcal{H}%
_{\mathcal{E}}$ of all functions $v$ for which (\ref{Eq5.3}) holds.

The triangle inequality
\begin{equation}
\mathcal{E}_{c}\left(  u+v\right)  ^{\frac{1}{2}}\leq\mathcal{E}_{c}\left(
u\right)  ^{\frac{1}{2}}+\mathcal{E}_{c}\left(  v\right)  ^{\frac{1}{2}}
\label{Eq5.7}%
\end{equation}
holds; or equivalently
\[
\left\Vert u+v\right\Vert _{\mathcal{E}}\leq\left\Vert u\right\Vert
_{\mathcal{E}}+\left\Vert v\right\Vert _{\mathcal{E}},~\forall u,v\in
\mathcal{H}_{\mathcal{E}}\text{.}%
\]

In the next result we give a characterization of the Hilbert space
$\mathcal{H}_{\mathcal{E}\text{ }}$directly in terms of the selfadjoint
operator $\overline{\Delta_{c}}$ from section \ref{MainTheo}.

Recall $\overline{\Delta_{c}}$ is the closure of the operator $\Delta_{c}$
with dense domain $\mathcal{D}$ in $\ell^{2}\left(  G^{\left(  0\right)
}\right)  $. It will be convenient to write simply $\Delta_{c}$ for the
closure. Since it is selfadjoint, we have a Borel functional calculus; i.e.,
if $f$ is a Borel function on $\mathbb{R}$, and if $P\left(  \cdot\right)  $
is a projection valued measure for $\Delta_{c}$, then
\begin{equation}
\Delta_{c}=\int_{0}^{\infty}\lambda P\left(  d\lambda\right)  \text{,}
\label{Eq5.8}%
\end{equation}
and%
\begin{equation}
f\left(  \Delta_{c}\right)  :\,=\int f\left(  \lambda\right)  P\left(
d\lambda\right)  \text{.} \label{Eq5.9}%
\end{equation}

For the corresponding (dense) domains, we have
\begin{equation}
\operatorname*{dom}\left(  \Delta_{c}\right)  =\left\{  v\in\ell^{2}\left(
G^{\left(  0\right)  }\right)  |\int_{0}^{\infty}\left\vert \lambda\right\vert
^{2}\left\Vert P\left(  d\lambda\right)  v\right\Vert ^{2}<\infty\right\}
\text{,} \label{Eq5.10}%
\end{equation}
and
\begin{equation}
\operatorname*{dom}\left(  f\left(  \Delta_{c}\right)  \right)  =\left\{
v\in\ell^{2}\left(  G^{\left(  0\right)  }\right)  |\int_{0}^{\infty
}\left\vert f\left(  \lambda\right)  \right\vert ^{2}\left\Vert P\left(
d\lambda\right)  v\right\Vert ^{2}<\infty\right\}  \text{.} \label{Eq5.11}%
\end{equation}

\begin{theorem}
\label{Theo5.1}Let $G=\left(  G^{\left(  0\right)  },G^{\left(  1\right)
}\right)  $ and $c:G^{\left(  1\right)  }\rightarrow\mathbb{R}_{+}$ be as
described, and let $\mathcal{H}_{\mathcal{E}}$ be the energy Hilbert space.
Let $\Delta_{c}$ be the selfadjoint graph Laplacian in $\ell^{2}\left(
G^{\left(  0\right)  }\right)  $ from Section \ref{MainTheo}.\smallskip

\emph{(}a\emph{)} Then
\begin{equation}
\operatorname*{dom}\left(  \Delta_{c}^{1/2}\right)  =\mathcal{H}_{\mathcal{E}%
}\cap\ell^{2}\left(  G^{\left(  0\right)  }\right)  \text{.} \label{Eq5.12}%
\end{equation}

\emph{(}b\emph{)} In general the right hand side is a proper subspace of
$\ell^{2}\left(  G^{\left(  0\right)  }\right)  $.
\end{theorem}

\begin{proof}
We proved in (\ref{Eq3.11}), Lemma \ref{Lem3.2}, that
\begin{equation}
\mathcal{E}_{c}\left(  u,v\right)  =2\langle u,\Delta_{c}v\rangle
\label{Eq5.13}%
\end{equation}
for all functions $u,v$ on $G^{\left(  0\right)  }$ for which the two sides in
(\ref{Eq5.13}) are finite. If $u=v$, then the expression on the RHS in
(\ref{Eq5.13}) is $2\left\Vert \Delta_{c}^{1/2}v\right\Vert ^{2}$ iff
$v\in\operatorname*{dom}\left(  \Delta_{c}^{1/2}\right)  $. This follows from
(\ref{Eq5.11}) applied to
\[
f\left(  \lambda\right)  :\,=\sqrt{\lambda},~\lambda\in\lbrack0,\infty
)\text{.}%
\]
Since $\operatorname*{dom}\left(  \Delta_{c}\right)  \subset
\operatorname*{dom}\left(  \Delta_{c}^{1/2}\right)  \subset\ell^{2}\left(
G^{\left(  0\right)  }\right)  $ the desired conclusion (\ref{Eq5.12}) holds.

To see that $\operatorname*{dom}\left(  \Delta_{c}^{1/2}\right)  $ may be a
proper subspace of
\begin{equation}
\mathcal{H}_{\mathcal{E}}=\left\{  v|\mathcal{E}_{c}\left(  v\right)
<\infty\right\}  \text{,} \label{Eq5.14}%
\end{equation}
consider the following example $\left(  G,c\right)  $ built on the simplest
infinite graph \linebreak$G^{\left(  0\right)  }:\,=\mathbb{Z}$.
\end{proof}

\begin{example}
\label{Examp5.2}%
\[%
\begin{array}
[c]{l}%
G^{\left(  0\right)  }=\mathbb{Z}\text{,}\\
G^{\left(  1\right)  }=\left\{  \left(  n,n\pm1\right)  |n\in\mathbb{Z}%
\right\}  \text{, and}\\
c:G^{\left(  1\right)  }\rightarrow\mathbb{R}_{+},~c\equiv1\text{ on
}G^{\left(  1\right)  }\text{.}%
\end{array}
\]

The corresponding graph Laplacian is
\begin{equation}
\left(  \Delta v\right)  \left(  n\right)  =2v\left(  n\right)  -v\left(
n-1\right)  -v\left(  n+1\right)  \text{, }\forall n\in\mathbb{Z}\text{.}
\label{Eq5.15}%
\end{equation}
If $k\in\mathbb{Z}_{+}$ is given, we claim that there is a unique function
$v\in\mathcal{H}_{\mathcal{E}}$ solving
\begin{equation}
\Delta v=\delta_{0}-\delta_{k}\text{.} \label{Eq5.16}%
\end{equation}

\end{example}

\textit{Existence}: Set
\begin{equation}
v\left(  n\right)  :\,=\left\{
\begin{array}
[c]{ll}%
0 & \text{if }n\leq0\\
-n & \text{if }0<n\leq k\\
-k & \text{if }k\leq n\text{.}%
\end{array}
\right.  \label{Eq5.17}%
\end{equation}
A substitution shows that the function $v$ in (\ref{Eq5.17}) satisfies
(\ref{Eq5.16}).

\textit{Uniqueness}: Let $w\in\mathcal{H}_{\mathcal{E}}$ be a solution to
(\ref{Eq5.16}). Then
\[
\mathcal{E}_{c}\left(  u,v-w\right)  =0,~\forall u\in\mathcal{D}\text{.}%
\]
Since $G=\left(  G^{\left(  0\right)  },G^{\left(  1\right)  }\right)  $ is
connected, $\mathcal{D}$ is dense in $\mathcal{H}_{\mathcal{E}}$; and so the
difference $v-w$ must be a constant function. But the Hilbert space
$\mathcal{H}_{\mathcal{E}}$ is defined by moding out with the constants.
Hence, $v=w$ in $\mathcal{H}_{\mathcal{E}}$.

The following three facts follow directly from (\ref{Eq5.17}):\smallskip

(i) $v$ is non-constant;

(ii) $v\not \in \ell^{2}\left(  \mathbb{Z}\right)  $; and

(iii) $\mathcal{E}_{c}\left(  v\right)  <\infty$.\smallskip

\noindent In fact, an application of (\ref{Eq5.13}) yields
\begin{equation}
\mathcal{E}_{c}\left(  v\right)  =2k\text{.} \label{Eq5.18}%
\end{equation}

\noindent Proof of (\ref{Eq5.18}):
\begin{align*}
\mathcal{E}_{c}\left(  v\right)   &  =2\langle v,\Delta v\rangle\text{ by
(\ref{Eq5.13})}\\
&  =2\langle v,\delta_{0}-\delta_{k}\rangle\\
&  =2\left(  v\left(  0\right)  -v\left(  k\right)  \right) \\
&  =2k\text{ by (\ref{Eq5.17}).}%
\end{align*}

\begin{theorem}
\label{Theo5.3}Let $G=\left(  G^{\left(  0\right)  },G^{\left(  1\right)
}\right)  $ and $c:G^{\left(  1\right)  }\rightarrow\mathbb{R}_{+}$ be a graph
system as in the previous theorem, and in section \ref{MainTheo}; i.e., we
assume that the pair $\left(  G,c\right)  $ satisfies the axioms listed there.
Let $\Delta_{c}$ be the corresponding graph Laplacian with a choice $c$ for conductance.

Let $\alpha,\beta\in G^{\left(  0\right)  }$ be a fixed pair of vertices. Then
there is a unique function $v\in\mathcal{H}_{\mathcal{E}}$, i.e.,
$\mathcal{E}_{c}\left(  v\right)  <\infty$ satisfying
\begin{equation}
\Delta_{c}v=\delta_{\alpha}-\delta_{\beta}\text{.} \label{Eq5.19}%
\end{equation}
The solution to \emph{(}\ref{Eq5.19}\emph{)} is called a voltage potential.
Moreover,
\begin{equation}
\mathcal{E}_{c}\left(  v\right)  =2\left(  v\left(  \alpha\right)  -v\left(
\beta\right)  \right)  \text{.} \label{Eq5.20}%
\end{equation}

\end{theorem}

\begin{proof}
The argument for \textit{uniqueness} is the same as in the previous proof.

To prove \textit{existence}, we will appeal to Riesz' theorem for the energy
Hilbert space $\mathcal{H}_{\mathcal{E}}$.

Hence, we must show that there is a finite constant $K$ such that
\begin{equation}
\left\vert u\left(  \alpha\right)  -u\left(  \beta\right)  \right\vert \leq
K\mathcal{E}_{c}\left(  u\right)  ^{\frac{1}{2}}\text{ for all }%
u\in\mathcal{D}\text{.} \label{Eq5.21}%
\end{equation}

Motivated by Ohm's law, we set $\Omega\left(  e\right)  :\,=c\left(  e\right)
^{-1},~\forall e\in G^{\left(  1\right)  }$. By the assumptions in section
\ref{MainTheo}, we may pick a finite subset $x_{0},x_{1},x_{2},\ldots,x_{n}$
in $G^{\left(  0\right)  }$ such
\begin{equation}
\left\{
\begin{array}
[c]{l}%
x_{0}=\alpha,~x_{n}=\beta\text{, and }\\
e_{i}=\left(  x_{i}\,x_{i+1}\right)  \in G^{\left(  1\right)  },~i=0,1,\ldots
,n-1\text{.}%
\end{array}
\right.  \label{Eq5.22}%
\end{equation}
Then%
\begin{align*}
\left\vert u\left(  \alpha\right)  -u\left(  \beta\right)  \right\vert  &
\leq\sum_{i=0}^{n-1}\left\vert u\left(  x_{i}\right)  -u\left(  x_{i+1}%
\right)  \right\vert \smallskip\\
&  \leq\left(  \sum_{i=0}^{n-1}\Omega\left(  e_{i}\right)  \right)  ^{\frac
{1}{2}}\left(  \sum_{i=0}^{n-1}c\left(  e_{i}\right)  \left\vert u\left(
x_{i}\right)  -u\left(  x_{i+1}\right)  \right\vert ^{2}\right)  ^{\frac{1}%
{2}}~\text{{\small (by Schwarz)}}\\
&  \leq\left(  \sum_{i=1}^{n-1}\Omega\left(  e_{i}\right)  \right)  ^{\frac
{1}{2}}\mathcal{E}_{c}\left(  u\right)  ^{\frac{1}{2}}\text{.}%
\end{align*}

To get a finite constant $K$ in (\ref{Eq5.21}), we may take the infimum over
all paths subject to conditions (\ref{Eq5.22}), connecting $\alpha$ to $\beta$.

An application of Riesz' lemma to $\mathcal{H}_{\mathcal{E}}$ yields a unique
$v\in\mathcal{H}_{\mathcal{E}}$ such that for all $u\in\mathcal{D}$, we have
the following identity:%
\begin{align*}
u\left(  \alpha\right)  -u\left(  \beta\right)   &  =\frac{1}{2}%
\mathcal{E}_{c}\left(  v,u\right) \\
&  =\langle\Delta_{c}v,u\rangle~\text{{\small (by (\ref{Eq5.13})).}}%
\end{align*}
Using again density of $\mathcal{D}$ in $\mathcal{H}_{\mathcal{E}}$, we get
the desired conclusion
\begin{equation}
\Delta_{c}v=\delta_{\alpha}-\delta_{\beta}\text{.} \label{Eq5.23}%
\end{equation}

\end{proof}

\begin{corollary}
\label{Cor5.4}Let $\left(  G,c\right)  $ satisfy the conditions in the
theorem. Let $\alpha,\beta\in G^{\left(  0\right)  }$, and let $v\in
\mathcal{H}_{\mathcal{E}}$ be the solution \emph{(}potential\emph{)}\thinspace
to
\[
\Delta_{c}v=\delta_{\alpha}-\delta_{\beta}\text{.}%
\]
Then
\begin{equation}
\mathcal{E}_{c}\left(  v\right)  \leq2\inf\limits_{\left(  e_{i}\right)  }%
\sum_{i=0}^{n-1}\Omega\left(  e_{i}\right)  \label{Eq5.24}%
\end{equation}
where $e_{0},e_{1},\ldots,e_{n-1}\in G^{\left(  1\right)  }$ is a system of
edges connecting $\alpha$ to $\beta$, i.e., satisfying the conditions listed
in \emph{(}\ref{Eq5.22}\emph{)}.
\end{corollary}

\begin{proof}
This follows from the previous proof combined with the fact that
\begin{equation}
\sup\limits_{\mathcal{E}_{c}\left(  u\right)  =1}\left\vert \mathcal{E}%
_{c}\left(  u,v\right)  \right\vert ^{2}=\mathcal{E}_{c}\left(  v\right)
\text{.} \label{Eq5.25}%
\end{equation}
\medskip
\end{proof}

\subsection{\textbf{A matrix representation\label{MatrixRep}}}

While $\Delta_{c}$ may be understood as an operator, it is also an
$\infty\times\infty$ matrix. Since the set $\operatorname*{nbh}\left(
x\right)  \subset G^{\left(  0\right)  }$ is finite for all $x\in G^{\left(
0\right)  }$, $\Delta_{c}$ is a \textit{banded matrix}. To see this, note that
when $x\in G^{\left(  0\right)  }$ is fixed, the summation
\begin{equation}
\left(  \Delta_{c}v\right)  \left(  x\right)  =\sum_{y\sim x}c\left(
xy\right)  \left(  v\left(  x\right)  -v\left(  y\right)  \right)
\label{Eq5.26}%
\end{equation}
is finite for all functions $v:G^{\left(  0\right)  }\rightarrow\mathbb{C}$.

Since $G$ is assumed connected, the only bounded solutions $v$ to the
equation
\begin{equation}
\Delta_{c}v=0 \label{Eq5.27}%
\end{equation}
are the constants.

Solutions $v$ to (\ref{Eq5.27}) are called \textit{harmonic}, or $c$-harmonic.

There are examples of systems $\left(  G,c\right)  $ which are connected and
have unbounded non-constant harmonic functions, e.g., models with $G^{\left(
0\right)  }=\mathbb{Z}^{3}$, or tree-models.

In the general case, introducing
\begin{equation}
\mathcal{B}_{c}\left(  x\right)  :\,=\sum_{y\sim x}c\left(  xy\right)  ,~x\in
G^{\left(  0\right)  }\text{;} \label{Eq5.28}%
\end{equation}
we see that (\ref{Eq5.26}) takes the following form
\begin{equation}
\left(  \Delta_{c}v\right)  \left(  x\right)  =\mathcal{B}_{c}\left(
x\right)  v\left(  x\right)  -\sum_{y\sim x}c\left(  xy\right)  v\left(
y\right)  \text{.} \label{Eq5.29}%
\end{equation}

Hence eq. (\ref{Eq5.27}) may be rewritten as
\begin{equation}
v\left(  x\right)  =\frac{1}{\mathcal{B}_{c}\left(  x\right)  }\sum_{y\sim
x}c\left(  xy\right)  v\left(  y\right)  \text{.} \label{Eq5.29.5}%
\end{equation}
It follows that harmonic functions on $G^{\left(  0\right)  }$ satisfy a mean
value property. At every $x\in G^{\left(  0\right)  }$ formula (\ref{Eq5.29.5}%
) expresses $v\left(  x\right)  $ as a convex combination of its values on the
set $\operatorname*{nbh}\left(  x\right)  $.

In matrix language, $x\rightarrow\mathcal{B}_{c}\left(  x\right)  $ represents
the diagonal matrix-entries; and $c\left(  xy\right)  $ the off-diagonal
entries. Since $\{y\in G^{\left(  0\right)  }|c\left(  xy\right)  \not =0\}$
is finite, we say that the matrix for $\Delta_{c}$ is \textit{banded}. It is
clear that products of banded matrices are again banded; and in particular
that the summations involved in matrix-products of banded matrices are all
finite. Hence, each of the operators $\Delta_{c},\Delta_{c}^{2},\Delta_{c}%
^{3},\ldots$, is banded. Since by Theorem \ref{Theo3.1}, $\Delta_{c}$ is
selfadjoint as an operator in $\ell^{2}\left(  G^{\left(  0\right)  }\right)
$, the fractional power $\Delta_{c}^{1/2}$ is well defined by the Spectral
Theorem. The matrix-entries of $\Delta_{c}^{1/2}$ are the numbers
\begin{equation}
\langle\delta_{x},\Delta_{c}^{1/2}\delta_{y}\rangle_{\ell^{2}}=\left(
\Delta_{c}^{1/2}\delta_{y}\right)  \left(  x\right)  ,~x,y\in G^{\left(
0\right)  }\text{.} \label{Eq5.30}%
\end{equation}
It can be checked that if $G$ is infinite, the matrix for $\Delta_{c}^{1/2}$
is typically \textit{not} banded. The same conclusion applies to $\Delta
_{c}^{s}$ when $s\in\mathbb{R}\diagdown\mathbb{N}$.

\subsection{\textbf{Example 5.2 revisited\label{Ex5.2Revisit}}}

The system $\left(  G,c\right)  $ in Example \ref{Examp5.2} does not have
non-constant harmonic functions. This can be seen from the representation of
$\Delta$ (in Ex. \ref{Examp5.2}) as a $\mathbb{Z}\times\mathbb{Z}$ double
infinite matrix, i.e.,
\begin{align*}
\left(  \Delta v\right)  \left(  n\right)   &  =2v\left(  n\right)  -v\left(
n-1\right)  -v\left(  n+1\right) \\
&  =v\left(  n\right)  -v\left(  n-1\right)  +v\left(  n\right)  -v\left(
n+1\right) \\
&  =\sum_{m\sim n}v\left(  n\right)  -v\left(  m\right)  ,~n\in\mathbb{Z}.
\end{align*}
In matrix form, $\Delta$ from Example \ref{Examp5.2} is as follows:%

\[
\left(
\begin{array}
[c]{rrrrrrrrrrr}%
\ddots & \ddots & \ddots &  &  & \cdots &  &  &  & \vdots & \\
& -1 & 2 & -1 & 0 & 0 & 0 & 0 & 0 & 0 & \\
& 0 & -1 & 2 & -1 & 0 & 0 & 0 & 0 & 0 & \\
& 0 & 0 & -1 & 2 & -1 & 0 & 0 & 0 & 0 & \\
& 0 & 0 & 0 & -1 & 2 & -1 & 0 & 0 & 0 & \\
& 0 & 0 & 0 & 0 & -1 & 2 & -1 & 0 & 0 & \\
& 0 & 0 & 0 & 0 & 0 & -1 & 2 & -1 & 0 & \\
& 0 & 0 & 0 & 0 & 0 & 0 & -1 & 2 & -1 & \\
& \vdots &  &  &  & \cdots &  &  & \ddots & \ddots & \ddots
\end{array}
\right)
\]

Using Fourier series
\begin{equation}
f\left(  x\right)  =\sum_{n\in\mathbb{Z}\,}v\left(  n\right)  e^{inx}\in
L^{2}\left(  -\pi,\pi\right)  \text{;} \label{Eq5.31}%
\end{equation}%
\begin{equation}
\sum_{n\in\mathbb{Z}\,}\left\vert v\left(  n\right)  \right\vert ^{2}=\frac
{1}{2\pi}\int_{-\pi}^{\pi}\left\vert f\left(  x\right)  \right\vert
^{2}~dx\text{;} \label{Eq5.32}%
\end{equation}
we arrive at the representation
\begin{align}
\left(  \tilde{\Delta}f\right)  \left(  x\right)   &  =2\left(  1-\cos
x\right)  f\left(  x\right) \label{Eq5.33}\\
&  =4\sin^{2}\left(  \frac{x}{2}\right)  f\left(  x\right)  \text{, }\nonumber
\end{align}
proving that $\Delta$ has Lebesgue spectrum, and
\begin{equation}
\operatorname*{spec}\limits_{\ell^{2}}\left(  \Delta\right)
=\operatorname*{spec}\limits_{L^{2}}\left(  \tilde{\Delta}\right)  =\left[
0,4\right]  \text{.} \label{Eq5.34}%
\end{equation}

\subsection{\textbf{Banded Matrices (A Preview)}\label{BandedMat}}

It is immediate from the matrix representation for $\Delta_{c}$ in Example
\ref{Examp5.2} that it has a \textit{banded} form. We will take up banded
infinite matrices in detail in section \ref{Append} below.

Since $\Delta_{c}$ is selfadjoint, its square-root $\Delta_{c}^{1/2}$ is a
well defined operator. However its matrix representation is typically
\textit{not} banded; see (\ref{Eq5.30}). For $\Delta_{c}^{1/2}$ in Ex.
\ref{Examp5.2}, one can check that the $\left(  m,n\right)  $-matrix entries
are
\[
\left(  \Delta_{c}^{\frac{1}{2}}\right)  _{m,n}\simeq\frac{1}{4\cdot\left(
n-m\right)  ^{2}+1}\text{.}%
\]

\subsection{\textbf{Extended Hilbert Spaces\label{ExtHilbertSpaces}}}

To understand solutions $v$ to operator equations like
\[
\Delta_{c}v=\delta_{\alpha}-\delta_{\beta}%
\]
as in (\ref{Eq5.23}), potential functions it is convenient to extend the
Hilbert space $\ell^{2}\left(  G^{\left(  0\right)  }\right)  $. Indeed we saw
in Example \ref{Examp5.2} that the solutions $v$ to equations like
(\ref{Eq5.23}) are typically not in $\ell^{2}\left(  G^{\left(  0\right)
}\right)  $.

\begin{definition}
\label{Def5.5}The space $\mathcal{H}_{c}\left(  s\right)  $.
\end{definition}

A function $v:G^{\left(  0\right)  }\rightarrow\mathbb{C}$ is said to belong
to the space $\mathcal{H}_{c}\left(  s\right)  $ if there is a finite constant
$K=K\left(  s\right)  $ such that the following estimate holds:
\begin{equation}
\left\vert \sum_{x\in G^{\left(  0\right)  }}\overline{v\left(  x\right)
}\left(  \Delta_{c}^{s}u\right)  \left(  x\right)  \right\vert ^{2}\leq
K\left(  s\right)  \sum_{x\in G^{\left(  0\right)  }}\left\vert u\left(
x\right)  \right\vert ^{2}~\text{for all }u\in\mathcal{D}\text{.}
\label{Eq5.35}%
\end{equation}

If (\ref{Eq5.35}) holds, then by Riesz, there is a unique $w\in\ell^{2}\left(
G^{\left(  0\right)  }\right)  $ such that
\begin{equation}
\sum_{x\in G^{\left(  0\right)  }}\overline{v\left(  x\right)  }\left(
\Delta_{c}^{s}u\right)  \left(  x\right)  =\langle w,u\rangle_{\ell^{2}%
}~\text{for all }u\in\mathcal{D~}\text{(}\subset\ell^{2}\text{);}
\label{Eq5.36}%
\end{equation}
and we set
\begin{equation}
\left\Vert v\right\Vert _{\mathcal{H}_{c}^{\left(  s\right)  }}:\,=\left\Vert
w\right\Vert _{\ell^{2}\left(  G^{0}\right)  }\text{.} \label{Eq5.37}%
\end{equation}
By abuse of notation, we will write $\Delta_{c}^{s}v=w$ when $v\in
\mathcal{H}_{c}\left(  s\right)  $.

If two functions $v_{i}$ for $i=1,2$ are in $\mathcal{H}_{c}\left(  s\right)
$, and if $\Delta_{c}^{s}v_{i}=w_{i}\in\ell^{2}\left(  G^{\left(  0\right)
}\right)  $, we set
\begin{align}
\langle v_{1},v_{2}\rangle_{\mathcal{H}_{c}\left(  s\right)  }  &  :\,=\langle
w_{1},w_{2}\rangle_{\ell^{2}}\label{Eq5.38}\\
&  =\sum_{x\in G^{\left(  0\right)  }}\overline{w_{1}\left(  x\right)  }%
w_{2}\left(  x\right)  \text{.}\nonumber
\end{align}

\begin{remark}
\label{Rem5.5.5}We proved in section \ref{MainTheo} that
\[
\sum_{x\in G^{\left(  0\right)  }}\left(  \Delta_{c}u\right)  \left(
x\right)  =0~\text{for }\forall u\in\mathcal{D}\text{.}%
\]
Hence the constant function $v_{1}\left(  x\right)  \equiv1$ on $G^{\left(
0\right)  }$ is in $\mathcal{H}_{c}\left(  1\right)  $, and $\left\Vert
v_{1}\right\Vert _{\mathcal{H}_{c}\left(  1\right)  }=0$. Hence in considering
the extension spaces, we shall work modulo the constant functions on
$G^{\left(  0\right)  }$.
\end{remark}

\begin{theorem}
\label{Theo5.6}For every $s\in\mathbb{R}$, the space $\mathcal{H}_{c}\left(
s\right)  $ is a Hilbert space when equipped with the inner product
\emph{(}\ref{Eq5.38}\emph{)}, and the norm \emph{(}\ref{Eq5.37}\emph{)}.
\end{theorem}

\begin{proof}
The idea in the proof follows closely the construction of Sobolev spaces, by
analogy to the continuous case. The key step in the verification of
completeness of $\mathcal{H}_{c}\left(  s\right)  $ is the essential
selfadjointness of $\Delta_{c}$ as an operator in $\ell^{2}\left(  G^{\left(
0\right)  }\right)  $. As before, we use the same notation $\Delta_{c}$ for
the closure of $\Delta_{c}$, defined initially only on the subspace
$\mathcal{D}$ in $\ell^{2}\left(  G^{\left(  0\right)  }\right)  $. Formulas
(\ref{Eq5.9})--(\ref{Eq5.11}) above now allow us to define the selfadjoint
operator $\Delta_{c}^{s}$; and this operator is closed in the sense that its
graph is closed in $\ell^{2}\left(  G^{\left(  0\right)  }\right)  \times
\ell^{2}\left(  G^{\left(  0\right)  }\right)  $. The completeness of
$\mathcal{H}_{c}\left(  s\right)  $ now follows from this, and an application
of Riesz; see the estimate (\ref{Eq5.36}).
\end{proof}

\begin{corollary}
\label{Cor5.7}Let $\left(  G,c\right)  $ be an infinite graph, and let
$c:G^{\left(  1\right)  }\rightarrow\mathbb{R}_{+}$ be a conductance function
satisfying the axioms above. Let $\alpha,\beta\in G^{\left(  0\right)  }$, and
let $v:G^{\left(  0\right)  }\rightarrow\mathbb{C}$ be a solution to
$\Delta_{c}v=\delta_{\alpha}-\delta_{\beta}$; i.e., to \emph{(}\ref{Eq5.23}%
\emph{)}.

Assume $v\in\mathcal{H}_{\mathcal{E}}$. Then
\begin{equation}
v\in\mathcal{H}_{c}\left(  \frac{1}{2}\right)  \cap\mathcal{H}_{c}\left(
1\right)  \text{\emph{;}} \label{Eq5.39}%
\end{equation}
and we have
\begin{equation}
\left\Vert v\right\Vert _{1/2}^{2}=\frac{1}{2}\mathcal{E}_{c}\left(  v\right)
, \label{Eq5.40}%
\end{equation}
and
\begin{equation}
\left\Vert v\right\Vert _{1}^{2}=2\text{.} \label{Eq5.41}%
\end{equation}

\end{corollary}

\begin{proof}
To prove (\ref{Eq5.39}), we must check the \textit{a priori} estimate
(\ref{Eq5.35}) for $s=1/2$, and $s=1$:%

\[
\text{\textbf{Verification of (\ref{Eq5.35}) for }}\mathbf{s=1/2}%
\]

Let $v$ satisfy the stated conditions, and let $u\in\mathcal{D}$. Then
\begin{align*}
\left\vert \sum_{x\in G^{\left(  0\right)  }}v\left(  x\right)  \left(
\Delta_{c}^{\frac{1}{2}}u\right)  \left(  x\right)  \right\vert  &
=\left\vert \sum_{x\in G^{\left(  0\right)  }}v\left(  x\right)  \Delta
_{c}\Delta_{c}^{-\frac{1}{2}}u\left(  x\right)  \right\vert \\
&  =\frac{1}{2}\left\vert \mathcal{E}_{c}\left(  v,\Delta_{c}^{-\frac{1}{2}%
}u\right)  \right\vert \\
&  \leq\frac{1}{2}\mathcal{E}_{c}\left(  v\right)  ^{\frac{1}{2}}%
\mathcal{E}_{c}\left(  \Delta_{c}^{-\frac{1}{2}}u\right)  ^{\frac{1}{2}}\text{
{\small (Schwarz)}}\\
&  =\frac{1}{\sqrt{2}}\mathcal{E}_{c}\left(  v\right)  ^{\frac{1}{2}%
}\left\Vert u\right\Vert _{\ell^{2}\left(  G^{\left(  0\right)  }\right)
}\text{, }%
\end{align*}
where we used the identity
\begin{align*}
\mathcal{E}_{c}\left(  \Delta_{c}^{-\frac{1}{2}}u\right)   &  =2\left\Vert
u\right\Vert _{\ell^{2}\left(  G^{\left(  0\right)  }\right)  }^{2}\\
&  =2\sum_{x\in G^{\left(  0\right)  }}\left\vert u\left(  x\right)
\right\vert ^{2}%
\end{align*}
valid for $\forall u\in\mathcal{D}$.%

\[
\text{\textbf{Verification of (\ref{Eq5.35}) for }}\mathbf{s=1}%
\]
With $v$ and $u$ as before, we must estimate the summation:%

\begin{align*}
\left\vert \sum_{x\in G^{\left(  0\right)  }}v\left(  x\right)  \left(
\Delta_{c}u\right)  \left(  x\right)  \right\vert  &  =\left\vert \sum_{x\in
G^{\left(  0\right)  }}\left(  \Delta_{c}v\right)  \left(  x\right)  u\left(
x\right)  \right\vert \\
&  =\left\vert \sum_{x\in G^{\left(  0\right)  }}\left(  \delta_{\alpha
}\left(  x\right)  -\delta_{\beta}\left(  x\right)  \right)  u\left(
x\right)  \right\vert \\
&  =\left\vert u\left(  \alpha\right)  -u\left(  \beta\right)  \right\vert \\
&  \leq2\left\Vert u\right\Vert _{\ell^{2}\left(  G^{\left(  0\right)
}\right)  },~\forall u\in\mathcal{D}\text{.}%
\end{align*}

Once (\ref{Eq5.39}) has been checked, the exact formulas (\ref{Eq5.40}) and
(\ref{Eq5.41}) follow:

Firstly,
\begin{align*}
\left\Vert v\right\Vert _{\frac{1}{2}}^{2}  &  =\left\Vert \Delta_{c}%
^{\frac{1}{2}}v\right\Vert _{\ell^{2}}^{2}\\
&  =\langle\Delta_{c}^{\frac{1}{2}}v,\Delta_{c}^{\frac{1}{2}}v\rangle\\
&  =\frac{1}{2}\mathcal{E}_{c}\left(  v\right)  \text{;}%
\end{align*}
and secondly
\begin{align*}
\left\Vert v\right\Vert _{1}^{2}  &  =\left\Vert \Delta_{c}v\right\Vert
_{\ell_{2}}^{2}\\
&  =\left\Vert \delta_{\alpha}-\delta_{\beta}\right\Vert _{\ell^{2}}^{2}\\
&  =2\text{.}%
\end{align*}

\end{proof}

\begin{remark}
\label{Rem5.7}In conclusion \emph{(}\ref{Eq5.39}\emph{)}\thinspace in
Corollary \ref{Cor5.7} is not best possible. In fact, the optimal range of the
fraction $s$ for which the potentials $v$ are in $\mathcal{H}_{c}\left(
s\right)  $ may be computed explicitly in Example \ref{Examp5.2} and related
examples. Details in the next subsection.

In Example \ref{Examp5.2}, $G^{\left(  0\right)  }=\mathbb{Z},$ $G^{\left(
1\right)  }=\left\{  \left(  n,n\pm1\right)  |n\in\mathbb{Z}\right\}  $, and
$c\equiv1$. Let $k\in\mathbb{N}$. The graph Laplacian $\Delta$ is given in
formula \emph{(}\ref{Eq5.15}\emph{)}.

Let $v$ be the unique solution to the potential equation
\begin{equation}
\Delta v=\delta_{0}-\delta_{k}\text{.} \label{Eq5.42}%
\end{equation}

Then $v\in\mathcal{H}\left(  s\right)  $ if and only if $s>1/4$.
\end{remark}

\begin{proof}
Setting
\begin{equation}
v\left(  z\right)  =\sum_{n\in\mathbb{Z}}v_{n}z^{n}\text{, and }%
z=e^{ix}\text{,} \label{Eq5.43}%
\end{equation}
we get
\begin{equation}
v\left(  z\right)  =\frac{z\left(  z^{k}-1\right)  }{\left(  z-1\right)  ^{2}%
}\text{;} \label{Eq5.44}%
\end{equation}
and therefore
\begin{equation}
\left\vert v\left(  x\right)  \right\vert =\frac{\left\vert \sin\left(
\frac{kx}{2}\right)  \right\vert }{\sin^{2}\left(  x/2\right)  }\text{.}
\label{Eq5.45}%
\end{equation}
Since, in the spectral representation, the graph Laplacian $\Delta$ is
multiplication by $4\sin^{2}\left(  x/2\right)  $, the question:
\textquotedblleft For what exponents $s$ is
\begin{equation}
v\in\mathcal{H}\left(  s\right)  \text{?\textquotedblright} \label{Eq5.46}%
\end{equation}
is decided by the asymptotics near $x=0$ of the function $\left(  \Delta
^{s}v\right)  \left(  x\right)  $. Using (\ref{Eq5.45}), we see that
$\Delta^{s}v$ is in $L^{2}\left(  -\pi,\pi\right)  $ if and only if
$x^{2s-1}\in L^{2}$ near $x=0$; and this hold if and only if
\begin{equation}
s>\frac{1}{4} \label{Eq5.47}%
\end{equation}
as claimed.
\end{proof}

\subsection{\textbf{Lattice Models\label{LatticeMods}}}

\begin{example}
\label{Examp5.9}We proved that potential functions are often not in $\ell
^{2}\left(  G^{\left(  0\right)  }\right)  $, but in general the problem is
more subtle.

The setting is a follows:

\begin{itemize}
\item[ ] $G=\left(  G^{\left(  0\right)  },G^{\left(  1\right)  }\right)  ~$a
given graph;

\item[ ] $c:G^{\left(  1\right)  }\rightarrow\mathbb{R}^{+}$ a given
conductance function;

\item[ ] $\Delta_{c}=$ the corresponding graph Laplacian;

\item[ ] $\alpha,\beta\in G^{\left(  0\right)  }$ a fixed pair of vertices,
$\alpha\not =\beta$.
\end{itemize}

With this, we say that a function $v:G^{\left(  0\right)  }\rightarrow
\mathbb{R}$ is a \textit{potential} if
\begin{equation}
\Delta_{c}v=\delta_{\alpha}-\delta_{\beta}\text{.} \label{Eq5.48}%
\end{equation}

\end{example}

In the next result we show that lattice models $\mathbb{Z}^{D}$ with $D>2$
have $\ell^{2}$ potentials.

\subsection{\textbf{Preliminaries\label{Prelim}}}

By $\mathbb{Z}^{D}$ we mean the rank $D$-lattice of vertex points $n=\left(
n_{1},n_{2},\ldots,n_{D}\right)  $, $n_{i}\in\mathbb{Z}$, $i=1,2,\ldots,D$.
Every point $n\in\mathbb{Z}^{D}$ has $2D$ distinct nearest neighbors
\begin{equation}
\left(  n_{1},\ldots,n_{i}\pm1,~n_{i+1},\ldots,n_{D}\right)  \text{,}
\label{Eq5.49}%
\end{equation}
so $\operatorname*{nbh}\left(  n\right)  $ consists of these $2D$ points; and
$G^{\left(  1\right)  }$ is the corresponding set of edges. In the discussion
below, we pick the constant conductance $c\equiv1,$ i.e., a system of
unit-resistors arranged in nearest-neighbor configurations. See Fig. 1 for an
illustration of the simplest lattice configuration, $D=1,2,$ and $3$.
\begin{center}
\includegraphics[
natheight=0.959700in,
natwidth=4.800300in,
height=0.7306in,
width=3.5724in
]%
{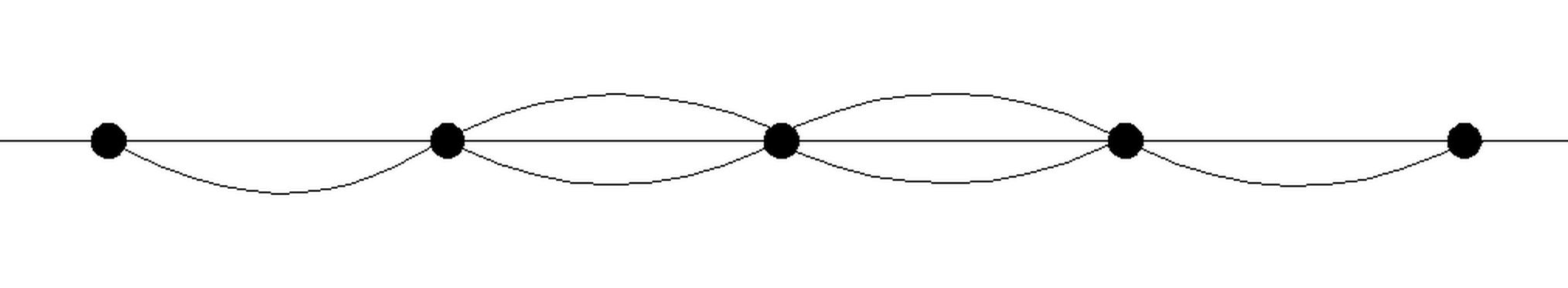}%
\\
Fig. 1a: $D=1$%
\label{Fig1a}%
\end{center}
%

\begin{center}
\includegraphics[
natheight=5.359900in,
natwidth=4.800300in,
height=2.7306in,
width=2.4475in
]%
{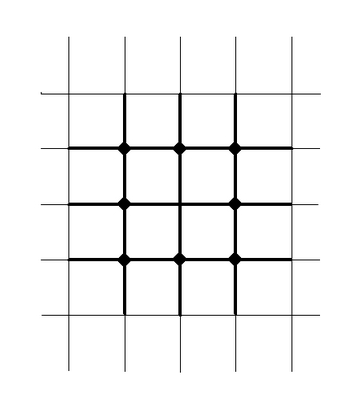}%
\\
Fig. 1b: $D=2$%
\label{Fig1b}%
\end{center}
%

\begin{center}
\includegraphics[
natheight=3.733500in,
natwidth=4.800300in,
height=2.166in,
width=2.7812in
]%
{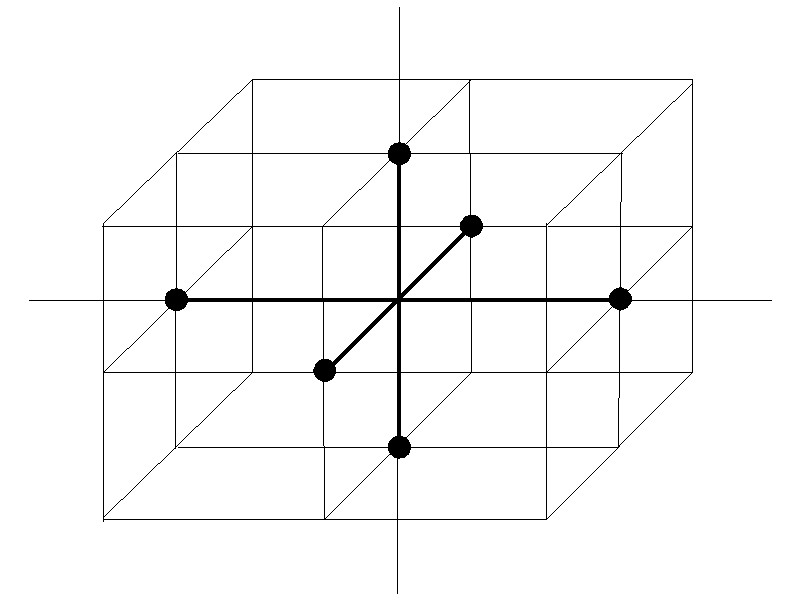}%
\\
Fig. 1c: $D=3$%
\label{Fig1c}%
\end{center}
%

\begin{align*}
&  \text{Fig. 1. Lattice configurations in the rank-}D\text{ lattices
}\mathbb{Z}^{D}\\
&  \text{with nearest-neighbor resistors.}%
\end{align*}

\begin{proposition}
\label{Prop5.10}The potential functions $v$, i.e., solutions to \emph{(}%
\ref{Eq5.48}\emph{)} with $c\equiv1$ are in $\ell^{2}(\mathbb{Z}^{D})$ if
$D>2$.
\end{proposition}

\begin{proof}
Recall that the $D$-torus $\mathbb{T}^{D}$ is the compact dual of the rank-$D$
lattice. Pick coordinates in $\mathbb{T}^{D}$ s.t. $x=\left(  x_{1}%
,\ldots,x_{D}\right)  $, $-\pi<x_{i}\leq\pi$, $i=1,2,\ldots,D$. Then, by
Parseval,
\[
\ell^{2}\left(  \mathbb{Z}^{D}\right)  \simeq L^{2}\left(  \mathbb{T}%
^{D}\right)  \text{.}%
\]
By the argument from Example \ref{Examp5.2}, we see that $\Delta$ has the
following spectral representation in $L^{2}(\underset{D\text{ times}%
}{\underbrace{(-\pi,\pi]\times\cdots\times(-\pi,\pi]}})$%
\begin{equation}
\left(  \Delta v\right)  \left(  x\right)  =4\sum_{k=1}^{D}\sin^{2}\left(
\frac{x_{k}}{2}\right)  v\left(  x\right)  \text{.} \label{Eq5.50}%
\end{equation}
Introducing spherical coordinates in $\mathbb{R}^{D}$, we get the
representation
\begin{equation}
dx=\rho^{D-1}dS_{1} \label{Eq5.51}%
\end{equation}
where%
\[
\rho:\,=\left(  \sum_{k=1}^{D}x_{k}^{2}\right)  ^{\frac{1}{2}}\text{,}%
\]
and whence $dS_{1}$ denotes the rotationally invariant measure on the sphere
in $\mathbb{R}^{D}$.

The question of deciding when the solution $v$ to (\ref{Eq5.48}) is in
$\ell^{2}\left(  \mathbb{Z}^{D}\right)  $ can be better understood in the
spectral representation $v\left(  x\right)  $ for $x=\left(  x_{1}%
,\ldots,x_{D}\right)  $ close to $0$, i.e., $\rho\left(  x\right)  \sim0$.

Using (\ref{Eq5.50})-(\ref{Eq5.51}), we see that the potential function $v$ is
in $\ell^{2}$ if $D>2$.

More generally, the argument from Example \ref{Examp5.2} proves that in
$\mathbb{Z}^{D}$, the potential function $v$ is in $\mathcal{H}\left(
s\right)  $ if $s>\frac{2-D}{4}$.\smallskip
\end{proof}

The next results illustrate new issues entering the analysis of $\mathbb{Z}%
^{D}$-graphs when $D>1$, compared to the $D=1$ case.

\begin{corollary}
\label{Cor5.11}For the case $D=3$ in the lattice model in Example
\ref{Examp5.9}, consider $k=\left(  k_{1},k_{2},k_{3}\right)  \in
\mathbb{Z}^{3}\diagdown\left(  0\right)  $ fixed, and let
\[
v_{k}:\mathbb{Z}^{3}\rightarrow\mathbb{R}\,
\]
the solution to the potential equation
\begin{equation}
\Delta v_{k}=\delta_{0}-\delta_{k}\text{.} \label{Eq5.52}%
\end{equation}
Then
\begin{equation}
\lim\limits_{n\rightarrow\infty}v_{k}\left(  n\right)  =0\text{.}
\label{Eq5.53}%
\end{equation}

\end{corollary}

\begin{proof}
Our notation is as follows: $n=(n_{1},n_{2},n_{3})\in\mathbb{Z}^{3}$ and by
\textquotedblleft$n\rightarrow\infty$\textquotedblright\ we mean:
\[
\left\vert n\right\vert =\sqrt{n_{1}^{2}+n_{2}^{2}+n_{3}^{2}}\rightarrow
\infty\text{.}%
\]
Moreover,
\[
\delta_{k}:\mathbb{Z}^{3}\rightarrow\mathbb{R}%
\]
is the usual Dirac mass
\begin{equation}
\delta_{k}\left(  n\right)  =\delta_{k_{1},n_{1}}\delta_{k_{i},n_{i}}%
\delta_{k_{3},n_{3}}\text{.} \label{Eq5.54}%
\end{equation}

We proved in Proposition \ref{Prop5.10} $\left(  D=3\right)  $ that
\begin{equation}
\sum_{n\in\mathbb{Z}^{3}}\left\vert v_{k}\left(  n\right)  \right\vert
^{2}\left(  =\left\Vert v_{k}\right\Vert ^{2}\right)  <\infty\text{;}
\label{Eq5.55}%
\end{equation}
and so in particular, the conclusion (\ref{Eq5.53}) must hold.
\end{proof}

Our next example illustrates that the potential equation (\ref{Eq5.52}) has
unbounded solutions in case $D\geq3$. This will also provide concrete cases of
unbounded harmonic functions, i.e., functions $w:\mathbb{Z}^{D}\rightarrow
\mathbb{R}$ for which $\Delta w=0$.

To aid the construction, we include the following lemma which is about the
general case of systems $\left(  G,c\right)  $ as analyzed in sections
\ref{MainTheo}-\ref{OpTheory} above.

\begin{lemma}
\label{Lem5.12}Let $G=(G^{\left(  0\right)  },G^{\left(  1\right)  })$, and
$c:G^{\left(  1\right)  }\rightarrow\mathbb{R}_{+}$, be a graph system as
described in Theorem \ref{Theo5.1}, and let $\Delta_{c}$ be the graph Laplacian.

Let $\alpha,\beta\in G^{\left(  0\right)  }$ be given, $\alpha\not =\beta$.
Then there is a 1-1 correspondence between two classes of functions
$v:G^{\left(  0\right)  }\rightarrow\mathbb{R}$, and functions $I:G^{\left(
1\right)  }\rightarrow\mathbb{R}$, where the two classes are given as
follows:\smallskip

\noindent\textit{Class 1.}%
\begin{equation}
\Delta_{c}v=\delta_{\alpha}-\delta_{\beta} \label{Eq5.56}%
\end{equation}

\noindent\textit{Class 2. }$I:G^{\left(  1\right)  }\rightarrow\mathbb{R}$
satisfying \emph{(}Kirchoff's Laws\emph{)}:

\begin{itemize}
\item[ ] \emph{(}a\emph{)} $\sum_{y\sim x}I\left(  x,y\right)  =\left(
\delta_{\alpha}-\delta_{\beta}\right)  \left(  x\right)  ,~\forall x\in
G^{\left(  0\right)  }$; and

\item[ ] \emph{(}b\emph{)} $\sum_{i}c\left(  x_{i}x_{i+1}\right)
^{-1}I\left(  x_{i}x_{i+1}\right)  =0$ for all $x_{0},x_{1},x_{2},\ldots
,x_{n}\in G^{\left(  0\right)  }$ subject to $x_{0}=x_{n}$, and $x_{i}\sim
x_{i+1}$, i.e., all closed loops in $G^{\left(  0\right)  }$.
\end{itemize}

The connection between the two classes is given by the following formula:
\begin{equation}
c\left(  xy\right)  ^{-1}I\left(  xy\right)  =v\left(  x\right)  -v\left(
y\right)  \text{,~}\forall\left(  xy\right)  \in G^{\left(  1\right)
}\text{.} \label{Eq5.57}%
\end{equation}

The function $v$ is determined from $I$ uniquely, up to a constant, when $I$
is known to satisfy \emph{(}a\emph{)}-\emph{(}b\emph{)}. Moreover,
\begin{equation}
\sum_{e\in G^{\left(  1\right)  }}\frac{\left(  I\left(  e\right)  \right)
^{2}}{c\left(  e\right)  }=\mathcal{E}_{c}\left(  v\right)  \text{.}
\label{Eq5.58}%
\end{equation}

\end{lemma}

\begin{proof}
Left to the reader. The arguments are included in the proof of Theorem
\ref{Theo5.1}.
\end{proof}

\begin{example}
\label{Examp5.13}The $\mathcal{D}=2$ lattice model; i.e., $G^{\left(
0\right)  }=\mathbb{Z}^{2}$, edges given by nearest neighbors as in Fig. 1b;
and $c\equiv1$.
\end{example}

We consider the equation (\ref{Eq5.56}) for $\alpha=\left(  0,0\right)  $ and
$\beta=\left(  1,1\right)  $. The two different solutions $v$ to
(\ref{Eq5.56}) will be presented in the form of Class 2 in Lemma
\ref{Lem5.12}, i.e., in terms of current functions defined on the edges in
$G$. First recall that the Laplace operator $\Delta$ in the $\mathbb{Z}^{2}%
$-model is
\begin{gather*}
\left(  \Delta v\right)  \left(  m,n\right)  =4v\left(  m,n\right)  -v\left(
m-1,n\right)  -v\left(  m+1,n\right) \\
-v\left(  m,n-1\right)  -v\left(  m,n+1\right)  ,~\forall\left(  m,n\right)
\in\mathbb{Z}^{2}\text{.}%
\end{gather*}
Eq. (\ref{Eq5.55}) then takes the form
\begin{equation}
\Delta v=\delta_{\left(  0,0\right)  }-\delta_{\left(  1,1\right)  }\text{.}
\label{Eq5.59}%
\end{equation}

We now describe the two current functions $I$ which correspond to the two
solutions to (\ref{Eq5.59}).%

\begin{gather*}%
\begin{array}
[c]{ccccccccccccccc}
&  &  &  &  & \bullet & \bullet & \bullet &  &  &  &  &  &  & \\
&  &  &  &  & \bullet & \bullet & \bullet &  &  &  &  &  &  & \\
&  &  &  &  &  &  &  &  &  &  &  &  &  & \\
& > &  & > &  & > &  & > &  & < &  & < &  & < & \\
\wedge &  & \wedge &  & \wedge &  & \wedge &  & \vee &  & \vee &  & \vee &  &
\\
& > &  & > &  & > &  & > &  & < &  & < &  & < & \\
\wedge &  & \wedge &  & \wedge &  & \wedge &  & \vee &  & \vee &  & \vee &  &
\vee\\
& > &  & > &  & > &  & > &  & < &  & < &  & < & \\
\wedge &  & \wedge &  & \wedge &  & \wedge &  & \vee &  & \vee &  & \vee &  &
\vee\\
& < &  & < &  & < &  & > &  & < &  & < &  & < & \\
\wedge &  & \wedge &  & \wedge &  & \wedge &  & \vee &  & \vee &  & \vee &  &
\vee\\
& < &  & < &  & < &  & > & \text{{\tiny (1,1)}} & < &  & < &  & < & \\
\wedge &  & \wedge &  & \wedge &  & \wedge & \therefore\because & \wedge &  &
\wedge &  & \wedge &  & \wedge\\
& < &  & < &  & < & \text{{\tiny (0,0)}} & > &  & > &  & > &  & > & \\
\vee &  & \vee &  & \vee &  & \vee &  & \vee &  & \wedge &  & \wedge &  &
\wedge\\
& < &  & < &  & < &  & > &  & > &  & > &  & > & \\
\vee &  & \vee &  & \vee &  & \vee &  & \vee &  & \wedge &  & \wedge &  &
\wedge\\
& < &  & < &  & < &  & > &  & > &  & > &  & > & \\
\vee &  & \vee &  & \vee &  & \vee &  & \vee &  & \wedge &  & \wedge &  &
\wedge\\
& < &  & < &  & < &  & > &  & > &  & > &  & > & \\
&  &  &  &  &  &  &  &  &  &  &  &  &  & \\
&  &  &  &  & \bullet & \bullet & \bullet &  &  &  &  &  &  & \\
&  &  &  &  & \bullet & \bullet & \bullet &  &  &  &  &  &  &
\end{array}
\\
\text{\textbf{Fig. 2.} The function }I\text{ for the first solution }v\text{
to (\ref{Eq5.59}).}%
\end{gather*}%
\begin{gather*}
\text{Flow design for the current function }I\text{. The symbols
\textquotedblleft%
$>$%
\textquotedblright\ indicate}\\
\text{arrows in the direction of the current flow. An arrow points in the}\\
\text{direction of voltage drop.}%
\end{gather*}

And now the (different) function $I$ for the second solution to (\ref{Eq5.59}%
):
\begin{gather*}%
\begin{array}
[c]{ccccccccccccccc}
&  & \vdots &  &  &  &  &  &  &  &  &  & \vdots &  & \\
&  & 0 &  & 0 &  & 0 &  & 0 &  & 0 &  & 0 &  & \\
& 0 &  & \overset{\wedge}{\frac{1}{2}} &  & 0 &  & 0 &  & 0 &  &
\underset{\vee}{\frac{1}{2}} &  & 0 & \\
&  & 0 &  & 0 &  & 0 &  & 0 &  & <\frac{1}{2} &  & <\frac{1}{2} &  & \\
& 0 &  & \overset{\wedge}{\frac{1}{2}} &  & 0 &  & 0 &  & \underset{\vee
}{\frac{1}{2}} &  & \underset{\vee}{\frac{1}{2}} &  & 0 & \\
&  & 0 &  & 0 &  & 0 &  & <\frac{1}{2} &  & <\frac{1}{2} &  & 0 &  & \\
\cdots & 0 &  & \overset{\wedge}{\frac{1}{2}} &  & 0 &  & \underset{\vee
}{\frac{1}{2}} &  & \underset{\vee}{\frac{1}{2}} &  & 0 &  & 0 & \\
&  & 0 &  & <\frac{1}{4} &  & <\frac{1}{2} &  & <\frac{1}{2} &  & 0 &  & 0 &
& \\
& 0 &  & \overset{\wedge}{\frac{1}{4}} &  & \underset{\vee}{\frac{1}{4}} &  &
\underset{\vee}{\frac{1}{2}} &  & 0 &  & 0 &  & 0 & \\
&  & 0 &  & \frac{1}{4}> & \text{{\tiny (1,1)}} & <\frac{1}{4} &  & 0 &  & 0 &
& 0 &  & \\
\cdots & 0 &  & \underset{\wedge}{\frac{1}{2}} &  & \overset{\wedge}{\frac
{1}{4}} &  & \underset{\vee}{\frac{1}{4}} &  & 0 &  & 0 &  & 0 & \\
&  & 0 & \text{{\tiny (0,0)}} & >\frac{1}{2} &  & \frac{1}{4}> &  & \frac
{1}{2}> &  & \frac{1}{2}> &  & \frac{1}{2}> &  & \frac{1}{2}>\\
& 0 &  & 0 &  & 0 &  & 0 &  & 0 &  & 0 &  & 0 & \cdots\\
&  & 0 &  & 0 &  & 0 &  & 0 &  & 0 &  & 0 &  & \\
&  & \vdots &  &  &  &  &  &  &  &  &  & \vdots &  &
\end{array}
\\
\text{\textbf{Fig. 3.}}%
\end{gather*}%
\begin{gather*}
\text{Flow design for the current function }I\text{. The symbols
\textquotedblleft%
$>$%
\textquotedblright\ indicate}\\
\text{arrows in the direction of the current flow. An arrow points in the}\\
\text{direction of voltage drop.}%
\end{gather*}

\subsection{\textbf{The Resistance Metric\label{ResistMet}}}

Let $G=(G^{\left(  0\right)  },G^{\left(  1\right)  })$ be a graph satisfying
the axioms from section \ref{Assumptions}, and let
\[
c:G^{\left(  1\right)  }\rightarrow\mathbb{R}_{+}%
\]
be a conductance function. Let $\mathcal{E}_{c}\left(  \cdot\right)  $ be the
corresponding energy form, and let $\Delta_{c}$ be the graph Laplacian.

Pick a definite point $0$ in the vertex set $G^{\left(  0\right)  }$. Now for
every $x\in G^{\left(  0\right)  }$ let $v_{x}\in\mathcal{H}_{\mathcal{E}_{c}%
}$ be the solution to
\begin{equation}
\Delta_{c}v_{x}=\delta_{0}-\delta_{x}\text{.} \label{Eq5.60}%
\end{equation}
Set
\begin{align}
\operatorname*{dist}_{c}\left(  x,y\right)   &  :\,=\mathcal{E}_{c}\left(
v_{x}-v_{y}\right)  ^{\frac{1}{2}}\label{Eq5.61}\\
&  =\left\Vert v_{x}-v_{y}\right\Vert _{\mathcal{E}_{c}}\nonumber
\end{align}
for $x,y\in G^{\left(  0\right)  }$. We say that $x,y\rightarrow
\operatorname*{dist}_{c}\left(  x,y\right)  $ is the \textit{resistance
metric} on $G^{\left(  0\right)  }$. It is immediate from (\ref{Eq5.60}) that
it satisfies the triangle inequality.

\begin{proposition}
\label{Prop5.14}The following formula holds for the resistance metric:
\[
\operatorname*{dist}\limits_{c}\left(  x,y\right)  =\sqrt{2}\left(
v_{x}\left(  y\right)  +v_{y}\left(  x\right)  -v_{x}\left(  x\right)
-v_{y}\left(  y\right)  \right)  ^{\frac{1}{2}}\text{.}%
\]

\end{proposition}

\begin{proof}
In view of (\ref{Eq5.61}), it is enough to compute $\mathcal{E}_{c}\left(
v_{x}-v_{y}\right)  $ for pairs of points $x,y$.

Let $x,y\in G^{\left(  0\right)  }$ be given, and let $v_{x},v_{y}$ be the
potential functions from (\ref{Eq5.60}). Then
\begin{align*}
\mathcal{E}_{c}\left(  v_{x}-v_{y}\right)   &  =2\langle\Delta_{c}\left(
v_{x}-v_{y}\right)  ,v_{x}-v_{y}\rangle_{\ell^{2}}\\
&  =2\langle\delta_{0}-\delta_{x}-\left(  \delta_{0}-\delta_{y}\right)
,v_{x}-v_{y}\rangle_{\ell^{2}}\\
&  =2\left(  \left(  v_{x}-v_{y}\right)  \left(  y\right)  -\left(
v_{x}-v_{y}\right)  \left(  x\right)  \right) \\
&  =2\left(  v_{x}\left(  y\right)  +v_{y}\left(  x\right)  -v_{x}\left(
x\right)  -v_{y}\left(  y\right)  \right)  \text{.}%
\end{align*}

\end{proof}

\begin{example}
\label{Examp5.14.1}\emph{(}See also Example \ref{Examp3.5.2}\emph{)} Let
$\Delta_{c}$ be given by the following $\infty\times\infty$ matrix:
\[
\left(
\begin{array}
[c]{rrrrlll}%
1 & -1\; & 0\; & 0\; & \cdots &  & \resizebox{!}{1cm}{\upshape{0}}\\
-1 & 5\; & -2^{2} & 0\; &  &  & \\
0 & -2^{2} & 2^{2}+3^{2} & -3^{2} &  &  & \\
\vdots &  & \ddots & \ddots & \ddots &  & \\
&  &  & \ddots & \ddots & \ddots & \\
&  &  &  & -n^{2} & n^{2}+\left(  n+1\right)  ^{2} & -\left(  n+1\right)
^{2}\\
&  &  &  &  & \ddots\qquad\quad\ddots & \\
& \resizebox{!}{1cm}{\upshape{0}} &  &  &  &  &
\end{array}
\right)  .
\]
So $G^{\left(  0\right)  }=\mathbb{N}_{0}$, $G^{\left(  1\right)  }=\left\{
\left(  0,1\right)  ,\cdots,\left(  n-1,n\right)  ,\left(  n,n+1\right)
,\cdots\right\}  $, and $c\left(  n,n+1\right)  =\left(  n+1\right)  ^{2}$.
The first vertex has one neighbor, and the later two.
\end{example}

The potential equation \emph{(}\ref{Eq5.61}\emph{)} may be solved by
inspection, and we get the following formula for the resistance metre
$\operatorname*{dist}\nolimits_{c}$ in Proposition \ref{Prop5.14}: If $m<n$
\emph{(}in $\mathbb{N}_{0}$\emph{)} then
\[
\operatorname*{dist}\nolimits_{c}\left(  m,n\right)  \simeq\left(  \frac
{1}{\left(  m+1\right)  ^{2}}+\frac{1}{\left(  m+2\right)  ^{2}}+\cdots
+\frac{1}{n^{2}}\right)  ^{\frac{1}{2}}\text{.}%
\]
Since $\sum_{k=1}^{\infty}\frac{1}{k^{2}}=\frac{\pi^{2}}{6}$, we conclude that
$\left(  G^{\left(  0\right)  },\operatorname*{dist}\nolimits_{c}\right)  $ is
a bounded metric space.

Further, the resistance is bounded at infinity; or equivalently the voltage
drop is \textquotedblleft very\textquotedblright\ slow at infinity for the
current flow induced by the experiment which inserts 1 amp at a particular
place in $G^{\left(  0\right)  }=\mathbb{N}_{0}$.

The reason is that the conductance is \textquotedblleft very\textquotedblright%
\ unbounded, or equivalently or more precisely, the resistance is
$\mathcal{O}\left(  n^{-2}\right)  $ for this particular $\left(  G,c\right)
$ system.

\noindent\textbf{Some conclusions: }The finite-energy solution $v$ to
(\ref{Eq5.59}) is the function $v:\mathbb{Z}^{2}\rightarrow\mathbb{R},$
beginning with the values $0,-1/2,$ and $-1$ as follows: In Figs. 2--3 we list
the values of $v$ on the points in the interior square in $G^{\left(
0\right)  }\left(  =\mathbb{Z}^{2}\right)  $. The three values are prescribed
in the centered square; and they then propagate into the quarter planes, with
the value $-1/2$ in the NW and the SE quarter planes.

\section{Finite Dimensional Approximation\label{^{FinDimApprox}}}

\subsection{\textbf{Systems of Graphs\label{GraphSystems}}}

Let $G=(G^{\left(  0\right)  },G^{\left(  1\right)  })$ be an infinite graph
satisfying the axioms from section 2. In particular, we assume for every $x$
in $G^{\left(  0\right)  }$ that $x$ itself is excluded from
$\operatorname*{nbh}(x)$; i.e., no $x$ in $G^{\left(  0\right)  }$ can be
connected to itself with a single edge. Let $c$ any conductance function
defined on $G^{\left(  1\right)  }$ and satisfying our usual axioms.

In section \ref{MainTheo} we showed that the corresponding Laplace operator
$\Delta=\Delta_{c}$ is automatically essentially selfadjoint. By this we mean
that when $\Delta$ is initially defined on the dense subspace $\mathcal{D}$
(of all the real valued functions on $G^{\left(  0\right)  }$ with finite
support) in the Hilbert space $\mathcal{H}:=\ell^{2}(G^{\left(  0\right)  })$,
then the closure of the operator $\Delta$ is selfadjoint in $\mathcal{H}$, and
so in particular it has a unique spectral resolution, determined by a
projection valued measure on the Borel subsets the infinite half-line
$\mathbb{R}_{+}$.

In contrast, we note (Example \ref{Examp7.0}) that the corresponding Laplace
operator in the continuous case is not essential selfadjoint.

This can be illustrated with $\Delta=-(d/dx)^{2}$ on the domain $\mathcal{D}$
of consisting of all $C^{2}$-functions on the infinite half-line
$\mathbb{R}_{+}$ which vanish with their derivatives at the end points. Then
the Hilbert space is $L^{2}\left(  \mathbb{R}_{+}\right)  $.

So this is an instance where the analogy between the continuous case and the
discrete case breaks down.

In the study of infinite graphs $G=\left(  G^{\left(  0\right)  },G^{\left(
1\right)  }\right)  $ and the corresponding Laplacians, it is useful to
truncate and consider first a nested system of \textit{finite} graphs $G_{N}$;
then compute in the finite case and, in the end, take the limit as
$N\rightarrow\infty$. Our approximation results here continue work started in
\cite{Jor77}, \cite{Jor78}.

\begin{definition}
\label{Def6.1}In this section we prove specific results showing that the
procedure works. While there are several candidates for designing the finite
approximating graphs $G_{N}=(G_{N}^{\left(  0\right)  },G_{N}^{\left(
1\right)  })$, we will concentrate here on the simplest: Starting with an
infinite $G=(G^{\left(  0\right)  },G^{\left(  1\right)  })$, pick finite
subsets of vertices as follows:
\begin{equation}
G_{1}^{\left(  0\right)  }\subset G_{2}^{\left(  0\right)  }\subset
G_{3}^{\left(  0\right)  }\subset\cdots\subset G_{N}^{\left(  0\right)
}\subset\cdots\subset G^{\left(  0\right)  }\label{Eq6.1}%
\end{equation}
such that
\begin{equation}
\bigcup\limits_{N=1}^{\infty}G_{N}^{\left(  0\right)  }=G^{\left(  0\right)
}\text{.}\label{Eq6.2}%
\end{equation}
Set $\mathcal{H}:\,=\ell^{2}(G^{\left(  0\right)  })$, and $\mathcal{H}%
_{N}=\ell^{2}(G^{\left(  0\right)  })$. Then the projection $P_{N}$ of
$\mathcal{H}$ onto $\mathcal{H}$ onto $\mathcal{H}_{N}$ is multiplication by
the indicator function $\chi_{G_{N}^{\left(  0\right)  }}$; and the projection
onto the complement $\mathcal{H}\circleddash\mathcal{H}_{N}$ is multiplication
with $\chi_{(G_{N}^{\left(  0\right)  })^{c}}$ where $(G_{N}^{\left(
0\right)  })^{c}=G^{\left(  0\right)  }\backslash G_{N}^{\left(  0\right)  }$
is the complement of $G_{N}^{\left(  0\right)  }$.

The edges $G_{N}^{\left(  1\right)  }$ in $G_{N}$ are simple the edges in $G$,
for which the vertices lie in $G_{N}^{\left(  0\right)  }$; i.e., if $x,y\in
G^{\left(  0\right)  }$, then:
\begin{equation}
\left(  xy\right)  \in G_{N}^{\left(  1\right)  }\Leftrightarrow\left(
xy\right)  \in G^{\left(  1\right)  }\text{ and }x,y\in G_{N}^{\left(
0\right)  }\text{.} \label{Eq6.3}%
\end{equation}

If a system $(G_{N})_{N\in\mathbb{N}}$ of graphs is given as in \emph{(}%
\ref{Eq6.1}\emph{)}-\emph{(}\ref{Eq6.3}\emph{)}, and if $c:G^{\left(
1\right)  }\rightarrow\mathbb{R}_{+}$ is a conductance function; we denote by
$c_{N}$ the restriction of $c$ to $G_{N}^{\left(  1\right)  }$.
\end{definition}

\begin{lemma}
\label{Lem6.2}Let $G=(G^{\left(  0\right)  },G^{\left(  1\right)  })$ and
$c:G^{\left(  1\right)  }\rightarrow\mathbb{R}_{+}$ be given as above. Let
$G_{N}$ be a system of graphs determined subject to conditions \emph{(}%
\ref{Eq6.1}\emph{)}-\emph{(}\ref{Eq6.3}\emph{)}.

Let $\Delta_{N}$ be the graph Laplacian associated to $(G_{N},c_{N})$. Then
\begin{equation}
P_{N}\Delta P_{N}=\Delta_{N},~\text{for }\forall N\in\mathbb{N}\text{.}
\label{Eq6.4}%
\end{equation}

\end{lemma}

\begin{proof}
For $v\in\mathcal{D}=$ finite linear combinations of $\{\delta_{x}|x\in
G^{\left(  0\right)  }\}$, we have
\begin{align*}
\left(  P_{N}\Delta P_{N}v\right)  \left(  x\right)   &  =\chi_{G_{N}}\left(
x\right)  \sum_{y\sim x}c\left(  xy\right)  \left(  \left(  \chi_{G_{N}%
}v\right)  \left(  x\right)  -\left(  \chi_{G_{N}}v\right)  \left(  y\right)
\right) \\
&  =\sum_{y\sim x\text{ in }G_{N}}c_{N}\left(  xy\right)  \left(  v\left(
x\right)  -v\left(  y\right)  \right) \\
&  =\left(  \Delta_{N}v\right)  \left(  x\right)  \text{;}%
\end{align*}
proving the formula (\ref{Eq6.4}).
\end{proof}

\begin{lemma}
\label{Lem6.3}Let $G=(G^{\left(  0\right)  },G^{\left(  1\right)  })$, and
$c:G^{\left(  1\right)  }\rightarrow\mathbb{R}_{+}$, be as in Lemma
\ref{Lem6.2} and Definition \ref{Def6.1}. Then for all $v\in\mathcal{D}$ and
$x\in G^{\left(  0\right)  }$, we have the following formula for the
difference operator $\Delta-\Delta_{N},~N=1,2,\ldots:$
\begin{equation}
\left(  \Delta v\right)  \left(  x\right)  -\left(  \Delta_{N}v\right)
\left(  x\right)  =-\chi_{G_{N}^{c}}\left(  x\right)  \sum_{%
\genfrac{}{}{0pt}{}{y\sim x}{y\in G_{N}}%
}c\left(  xy\right)  v\left(  y\right)  \text{.} \label{Eq6.5}%
\end{equation}
In other words, the contribution to $\Delta-\Delta_{N}$ comes from the
boundary of $G_{N}=$ the edges $e\in G^{\left(  1\right)  }$ s.t. one vertex
in $e$ is in $G_{N}^{\left(  0\right)  }$ and the other in the complement.
\end{lemma}

\begin{proof}
Using the previous lemma, we get
\begin{align*}
\left(  \Delta v\right)  \left(  x\right)  -\left(  \Delta_{N}v\right)
\left(  x\right)   &  =\sum_{y\sim x}\left(  c\left(  xy\right)  -c_{N}\left(
xy\right)  \right)  \left(  v\left(  x\right)  -v\left(  y\right)  \right) \\
&  =-\chi_{G_{N}^{c}}\left(  x\right)  \sum_{%
\genfrac{}{}{0pt}{}{y\sim x}{y\in G_{N}}%
}c\left(  xy\right)  v\left(  y\right)  \text{.}%
\end{align*}

\end{proof}

\begin{definition}
\label{Def6.4}Let $G=(G^{\left(  0\right)  },G^{\left(  1\right)  })$, and
$c:G^{\left(  1\right)  }\rightarrow\mathbb{R}_{+}$ be given as in Theorem
\ref{Theo5.1}; and denote by $\Delta=\Delta_{c}$ the corresponding
\emph{selfadjoint} graph Laplacian. Setting
\begin{align}
S\left(  t\right)   &  :\,=\int_{0}^{\infty}e^{-t\lambda}P\left(
d\lambda\right)  \quad(\text{see }(\text{\ref{Eq5.9}})\text{-}%
(\text{\ref{Eq5.10}}))\label{Eq6.6}\\
&  =e^{-t\Delta},~t\in\mathbb{R}_{+};\nonumber
\end{align}
we see that $t\longmapsto S\left(  t\right)  $ is a contractive semigroup of
selfadjoint operators in $\ell^{2}(G^{\left(  0\right)  })$; in particular,
\begin{align}
S\left(  s+t\right)   &  =S\left(  s\right)  S\left(  t\right)  ,~\forall
s,t\in\mathbb{R}_{+}\text{ and}\label{Eq6.7}\\
S\left(  0\right)   &  =I_{\ell^{2}}\text{.}\nonumber
\end{align}

The semigroup consists of \emph{bounded} operators while the infinitesimal
generator $\Delta=\Delta_{c}$ is typically unbounded, albeit with dense domain
in $\ell^{2}(G^{\left(  0\right)  })$. Moreover, the semigroup helps us
identify dynamics as infinite graphs of resistors.

Returning to approximations, as in Definition \ref{Def6.1}, we now get a
sequence of Laplacians $\Delta_{N}$, $N=1,2,\ldots$, and a corresponding
sequence of dynamical semigroups, $S_{N}\left(  t\right)  :\,=e^{-t\Delta_{N}%
},~N=1,2,\ldots$.

Let $N$ be fixed, and let $\partial G_{N}$ be the boundary of $G_{N}$
\emph{(}Definition \ref{Def6.1}\emph{)}. Then the finite matrix
\begin{equation}
T_{N}:\,=\left(  c\left(  xy\right)  \right)  _{x,y\in\partial G_{N}%
}\label{Eq6.8}%
\end{equation}
is positive, and has a Perron-Frobenius eigenvalue $\lambda_{N}=\lambda
_{N}(PF)=$ the spectral radius of $T_{N}$.
\end{definition}

\begin{theorem}
\label{Theo6.5}Let $(G,c)$ be a graph/conductance system, and let
$(G_{N})_{N\in\mathbb{N}}$ ascending system of graphs such that \emph{(}%
\ref{Eq6.2}\emph{)} is satisfied. Let $S\left(  t\right)  ,$ and $S_{N}\left(
t\right)  $, $N=1,2,\ldots$, be the corresponding semigroups of bounded operators.

Then for all $v\in\ell^{2}(G^{\left(  0\right)  })$, we have the following
estimate:
\begin{equation}
\left\Vert S\left(  t\right)  v-S_{N}\left(  t\right)  v\right\Vert _{\ell
^{2}}\leq\lambda_{N}\left(  PF\right)  t\left\Vert v\right\Vert _{\ell^{2}%
},~\forall t\in\mathbb{R}_{+},~N=1,2,\ldots\text{.} \label{Eq6.9}%
\end{equation}

\end{theorem}

\begin{proof}
With the use of \emph{(}\ref{Eq5.8}\emph{)}-\emph{(}\ref{Eq5.9}\emph{)}, we
get the integral formula:
\begin{equation}
e^{-t\Delta_{N}}-e^{-t\Delta}=\int_{0}^{t}e^{-\left(  t-s\right)  \Delta
}\left(  \Delta-\Delta_{N}\right)  e^{-s\Delta_{N}}~ds\text{.} \label{Eq6.10}%
\end{equation}

Since the operators on both sides in (\ref{Eq6.10}) are bounded, it is enough
to verify the estimate (\ref{Eq6.9}) for vectors $v$ in the dense domain
$\mathcal{D}$.

Using new Lemma \ref{Lem6.3}, we get the following estimates on the respective
$\ell^{2}$-norms:
\begin{align*}
\left\Vert S\left(  t\right)  v-S_{N}\left(  t\right)  v\right\Vert _{\ell
^{2}} &  \leq\int_{0}^{t}\left\Vert \left(  \Delta-\Delta_{N}\right)
S_{N}\left(  s\right)  v\right\Vert _{\ell^{2}}\text{ (by (\ref{Eq6.10}))}\\
&  \leq\lambda_{N}\left(  PF\right)  \int_{0}^{t}\left\Vert S_{N}\left(
s\right)  v\right\Vert ~ds\text{ (by Lemma \ref{Lem6.3} and (\ref{Eq6.8}))}\\
&  \leq\lambda_{N}\left(  PF\right)  \left\Vert v\right\Vert _{\ell^{2}}%
\int_{0}^{t}~ds\\
&  =\lambda_{N}\left(  PF\right)  t\left\Vert v\right\Vert _{\ell^{2}}\text{,}%
\end{align*}
which is the desired conclusion.
\end{proof}

\subsection{\textbf{Periodic boundary conditions\label{PerBoundCond}}}

\begin{example}
\label{Examp6.6}We now compare Example \ref{Examp5.2} with an associated
family of finite graphs $G_{N}$ where $N\in\mathbb{N}$. Let $\mathbb{Z}%
_{N}=\mathbb{Z}/N\mathbb{Z}\simeq\left\{  0,1,2,\ldots,N-1\right\}  $ be the
cyclic group of order $N$. Introduce nearest neighbors as in Example
\ref{Examp5.2} \emph{(}the $\mathbb{Z}$-case\emph{)} with the modification for
$G_{N}$ given by $0\sim\left(  N-1\right)  $, in other words that there is an
edge connecting $0$ to $N-1$.
\end{example}

It follows that the graph Laplacian $\Delta_{N}$ for $G_{N}$ is the given by
the finite matrix
\[
\left[
\begin{array}
[c]{rrrrrrrr}%
2 & -1 & 0 & 0 & 0 & \cdots & 0 & -1\\
-1 & 2 & -1 & 0 & 0 & \cdots & 0 & 0\\
0 & -1 & 2 & -1 & 0 & \cdots & 0 & 0\\
\vdots &  &  &  &  &  &  & \vdots\\
0 & 0 & \cdots &  &  & 2 & -1 & 0\\
0 & 0 & \cdots &  &  & -1 & 2 & -1\\
-1 & 0 & \cdots &  & \cdots & 0 & -1 & 2
\end{array}
\right]  \text{.}%
\]

The spectrum of $\Delta_{N}$ is as follows:
\begin{align}
\operatorname*{spec}\left(  \Delta_{N}\right)   &  =\left\{  2\left(
1-\cos\left(  \frac{2\pi k}{N}\right)  \right)  |k=0,1,\ldots,N-1\right\}
\label{Eq6.11}\\
&  =\left\{  4\sin^{2}\left(  \frac{\pi k}{N}\right)  |k=0,1,\ldots
,N-1\right\}  \text{.}\nonumber
\end{align}
Comparing with (\ref{Eq5.33})-(\ref{Eq5.34}), we see that the spectra converge
in a natural sense; with the infinite model in Ex \ref{Examp5.2} being a limit
of $N$-periodic boundary condition as $N\rightarrow\infty$.%

\begin{center}
\includegraphics[
natheight=6.000000in,
natwidth=8.000000in,
height=3.604in,
width=3.4852in
]%
{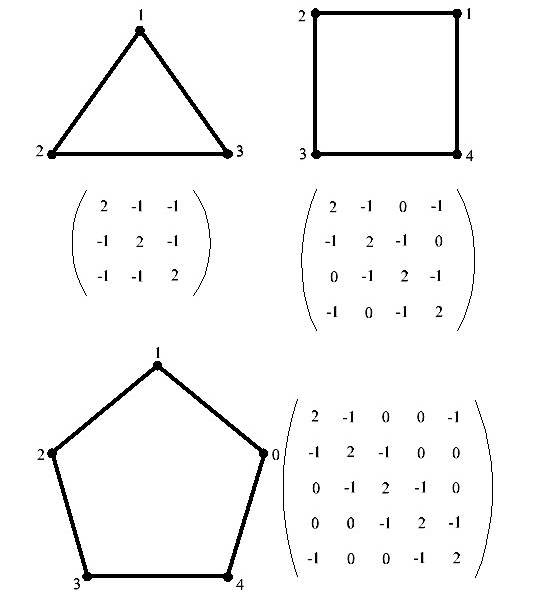}%
\\
\textbf{Fig. 4. }The vertices and edges in $G_{N}$ for $N=3,4$ and $5$.
\label{Fig4}%
\end{center}

The spectrum of the cyclic graph Laplacian $\Delta_{N}$ of the graphs $G_{N}$,
$N=3,4$ and $6$ (in Fig. 3) may have multiplicity; see (\ref{Eq6.11}). This
holds in fact for all values of $N$. Specifically,

\begin{itemize}
\item[ ] $\operatorname*{spec}\left(  \Delta_{3}\right)  =\left\{
0,3\right\}  $ with $\lambda=3$ having multiplicity $2$.

\item[ ] $\operatorname*{spec}\left(  \Delta_{4}\right)  =\left\{
0,2,4\right\}  $ with $\lambda=2$ having multiplicity $2$.

\item[ ] $\operatorname*{spec}\left(  \Delta_{6}\right)  =\left\{
0,1,3,4\right\}  $ now with $\lambda=1$ and $\lambda=3$ each having
multiplicity $2$.
\end{itemize}

Hence for $D=1$, we get the following distinction between the spectral theory
of the cyclic graph Laplacians $\Delta_{N}$ for $N<\infty$ on the one hand and
$\Delta$ in Example \ref{Examp5.2} on the other: The commutant of $\Delta$ is
an abelian algebra of operators in $\ell^{2}\left(  \mathbb{Z}\right)  $,
while the operators in $\ell^{2}\left(  \mathbb{Z}_{N}\right)  $ which commute
with $\Delta_{N}$ form a non-abelian algebra.

\begin{proposition}
\label{Prop6.7}\emph{(Cyclic graphs)}\thinspace Let $N\in\mathbb{N}$, $N\geq
3$; and let $G_{N}$ be the corresponding cyclic graph with graph Laplacian
$\Delta_{N}$; i.e., with
\[
G_{N}^{\left(  0\right)  }=\left\{  0,1,2,\ldots,N-1\right\}  .
\]
Then the voltage potential $v\in\ell^{2}\left(  \mathbb{Z}_{N}\right)  $
solving $\Delta_{N}v=\delta_{0}-\delta_{1}$is
\[
\left\{
\begin{array}
[c]{l}%
v_{0}=0\\
v_{1}=-\frac{N-1}{N}\\
v_{2}=-\frac{N-2}{N}\\
\multicolumn{1}{c}{\vdots}\\
v_{N-2}=-\frac{2}{N}\\
v_{N-1}=-\frac{1}{N}\text{.}%
\end{array}
\right.
\]

\end{proposition}

\begin{proof}
A direct computation; see also Fig. 4, and eq. (\ref{Eq6.10}).
\end{proof}

\section{Boundary Conditions\label{BoundaryCond}}

In the study of infinite graphs $G$, boundary conditions play an important
role; for example if a current escapes to infinity in \textquotedblleft finite
time,\textquotedblright\ conditions must then be assigned \textquotedblleft at
infinity.\textquotedblright

One way to do this is to first do computations in a system of finite graphs
$G_{N}$ which exhausts the given graph $G$ in a suitable way. Do computations
on each finite subgraph $G_{N}$ of the fixed infinite graph $G$, and then take
the limit as $N$ tends to infinity. There are several ways one may do the
computations on each individual $G_{N}$, for example look for symmetry, or
look for a suitable periodicity, or similarity up to scale. In the simplest
cases, this allows the use of a finite Fourier transform, thus making $G_{N}$
periodic, or cyclic. The case of $G=\mathbb{Z}$ (the rank-1 integer graph),
and $G_{N}=$ the cyclic group of order $N$ is done in all detail in Example
\ref{Examp6.6} above.

Some advantages of the cyclic approach: One, the spectrum comes out given
explicitly by a closed formula, thus making it clear how the limit
$N\rightarrow\infty$ works also for spectra, getting the continuous spectrum
in the infinite limit.

\begin{example}
\label{Examp7.0}In this section we compare the two cases, continuous vs.
discrete. As noted, our graph Laplacians are second order (or more than second
order) difference operators in a generalized sense.
\end{example}

They have spectrum contained in the half-line $[0,\infty)$, so generalizing
\begin{equation}
\left(  \Delta v\right)  \left(  x\right)  :\,=-\left(  \frac{d}{dx}\right)
^{2}v\left(  x\right)  \label{Eq7.1}%
\end{equation}
with the Hilbert space $\mathcal{H}:\,=L^{2}\left(  0,\infty\right)  $, and
dense domain
\begin{equation}
\mathcal{D}:\,=\left\{  v\in C^{2}\left(  0,\infty\right)  |v,v^{\prime
},v^{\prime\prime}\in L^{2}\left(  0,\infty\right)  ,\text{ and }v\left(
0\right)  =v^{\prime}\left(  0\right)  =0\right\}  \text{;} \label{Eq7.2}%
\end{equation}
i.e., with vanishing boundary conditions on $v$ and $v^{\prime}\left(
x\right)  =\frac{dv}{dx}$ at $x=0$.

We get the spectral estimate:
\begin{equation}
\langle v,\Delta v\rangle_{L^{2}}\geq0,\text{~}\forall v\in\mathcal{D}\text{.}
\label{Eq7.3}%
\end{equation}

A simple verification shows that for the adjoint operator $\Delta^{\ast}$ we
have:
\begin{equation}
\operatorname*{dom}\left(  \Delta^{\ast}\right)  =\left\{  v\in L^{2}\left(
0,\infty\right)  |v^{\prime},v^{\prime\prime}\in L^{2}\left(  0,\infty\right)
\right\}  \text{.} \label{Eq7.4}%
\end{equation}

Comparing (\ref{Eq7.2}) and (\ref{Eq7.4}) we see that $\Delta$ results from
$\Delta^{\ast}$ by \textquotedblleft removing\textquotedblright\ the two
boundary conditions which specify the domain $\mathcal{D}$ of $\Delta$.

Moreover, the defect space
\begin{equation}
\mathcal{D}_{+}:\,=\left\{  v\in\operatorname*{dom}\left(  \Delta^{\ast
}\right)  |\Delta^{\ast}v=-v\right\}  \label{Eq7.5}%
\end{equation}
is one-dimensional; in fact,
\begin{equation}
\mathcal{D}_{+}=\mathbb{C}e^{-x}\text{.} \label{Eq7.6}%
\end{equation}

The selfadjoint extensions of $\Delta$ on $\mathcal{D}$ are parametrized by
pairs of numbers $A,B\in\mathbb{R}$, not both zero, such that
\begin{equation}
Av\left(  0\right)  +Bv^{\prime}\left(  0\right)  =0\text{.} \label{Eq7.7}%
\end{equation}

\begin{example}
\label{Examp7.1}Let $G=(G^{\left(  0\right)  },G^{\left(  1\right)  })$ be the
following graph generalizing the continuous example:
\begin{align}
G^{\left(  0\right)  }  &  :\,=\mathbb{N}_{0}=\left\{  0,1,2,3,\ldots\right\}
,\label{Eq7.8}\\
G^{\left(  1\right)  }  &  :\,=\left\{  \left(  01\right)  ,\left(
n,n\pm1\right)  ;~n\in\mathbb{N}\right\}  \text{.}\nonumber
\end{align}

Pick $\lambda>1$, and set
\begin{equation}
c\left(  n,n+1\right)  =\lambda^{n+1}\text{.} \label{Eq7.9}%
\end{equation}
Then the corresponding graph Laplacian is unbounded; and
\begin{equation}
\left(  \Delta v\right)  \left(  0\right)  =\lambda v_{0}-\lambda
v_{1}\text{;} \label{Eq7.10}%
\end{equation}%
\begin{equation}
\left(  \Delta v\right)  \left(  n\right)  =-\lambda^{n}v_{n-1}+\lambda
^{n}\left(  1+\lambda\right)  v_{n}-\lambda^{n+1}v_{n+1}\text{,}\forall
n\in\mathbb{N}\text{.} \label{Eq7.11}%
\end{equation}

For domain $\mathcal{D}$, we take all $v\in\ell^{2}\left(  \mathbb{N}%
_{0}\right)  $ s.t. $v_{n}=0$ except for a finite set of values of $n$. the
matrix representation of $\Delta$ is presented in Fig. 4:
\begin{gather*}
\left(
\begin{array}
[c]{cccccccc}%
\lambda & -\lambda & 0 & 0 & \cdots & 0 & 0 & \cdots\\
-\lambda & \lambda\left(  1+\lambda\right)  & -\lambda^{2} & 0 &  &  &  & \\
0 & -\lambda^{2} & \lambda^{2}\left(  1+\lambda\right)  & -\lambda^{3} &  &  &
& \\
0 &  &  &  &  &  &  & \\
\vdots &  &  &  & -\lambda^{n} &  &  & \\
0 &  &  & -\lambda^{n} & \lambda^{n}\left(  1+\lambda\right)  & -\lambda^{n+1}
&  & \cdots\\
0 &  &  &  & -\lambda^{n+1} &  & \ddots &
\end{array}
\right) \\
\text{Fig. 4.}%
\end{gather*}

By Parseval's formula, we have the isometric isomorphism $\ell^{2}\left(
\mathbb{N}_{0}\right)  =\mathcal{H}_{+}=$ the Hardy space of analytic
functions on $D=\{z\in\mathbb{C}$; $\left\vert z\right\vert <1\}$
\[
v\left(  z\right)  :\,=\sum_{n=0}^{\infty}v_{n}z^{n}\text{;}%
\]
and
\begin{equation}
\left\Vert v\right\Vert _{\mathcal{H}_{+}}^{2}=\sum_{n=0}^{\infty}\left\vert
v_{n}\right\vert ^{2}\text{.} \label{Eq7.12}%
\end{equation}

In the Hardy space representation we have
\begin{equation}
\left(  \Delta v\right)  \left(  z\right)  =\left(  1+\lambda\right)  v\left(
\lambda z\right)  -\lambda zv\left(  \lambda z\right)  -z^{-1}v\left(  \lambda
z\right)  \label{Eq7.13}%
\end{equation}
on the dense space of functions $v$ on $\mathbb{C}$ which extend analytically
to $D_{\lambda}:\,=\{z\in\mathbb{C};\left\vert z\right\vert <\lambda\}$.

We now show that there are no non-zero solutions to
\begin{equation}
\Delta_{\lambda}^{\ast}v=-v\text{,} \label{Eq7.14}%
\end{equation}
i.e., $v\in\operatorname*{dom}\left(  \Delta_{\lambda}^{\ast}\right)  $;
equivalently $\mathcal{D}_{+}\left(  \lambda\right)  =\left\{  0\right\}  $;
the defect space for the operator $\Delta_{\lambda}$ is trivial. So this is a
direct verification that $\Delta_{\lambda}$ is essentially selfadjoint; and
contrasting with \emph{(}\ref{Eq7.6}\emph{)} above.
\end{example}

To see this, combine (\ref{Eq7.13})\thinspace and (\ref{Eq7.14}). It follows
that every solution $v$ to (\ref{Eq7.14}) must have an infinite-product
representation given by
\begin{equation}
v\left(  z\right)  =\frac{\left(  z-1\right)  \left(  \lambda z-1\right)
}{\lambda z}v\left(  \lambda z\right)  \text{;} \label{Eq7.15}%
\end{equation}
and the limit of finite products as follows
\[
\frac{\left(  z-1\right)  \prod\limits_{k=1}^{n-1}\left(  \lambda
^{k}z-1\right)  ^{2}\left(  \lambda^{n}z-1\right)  }{z^{n}\lambda
^{\frac{n\left(  n+1\right)  }{2}}}\text{.}%
\]
These products do not have a non-zero representation consistent with the
isomorphism (\ref{Eq7.12}), and with (\ref{Eq7.12}).\bigskip


\section{Appendix\label{Append}}


\appendix

\section{Heisenberg's Infinite Banded Matrices}

We proved in sections 3 through 5 that in general, graph Laplacians
$\Delta_{c}$ are essentially selfadjoint operators in the $\ell^{2}$
sequence-Hilbert space. Recall that the axioms for our graph Laplacians
include the following given data: A graph $G=(G^{\left(  0\right)
},G^{\left(  1\right)  })$ and a fixed positive conductance function $c$
defined on the set of edges $G^{\left(  1\right)  }$. Every vertex $x$ of $G$
is connected to a finite set of neighbors in $G^{\left(  0\right)  }$. For
every fixed $x$ in $G^{\left(  0\right)  }$, this implies finiteness of the
set of $y$ in $G^{\left(  0\right)  }$ for which $c(xy)$ is nonzero. This
means in turn that the natural matrix representation of the operator
$\Delta_{c}$ is \textit{banded}; see section \ref{EnergyForm} for the
Definition. Note however that we place no boundedness restrictions on the
conductance function $c$.

Our proof of essentially selfadjoint for the operator $\Delta_{c}$ uses this
bandedness property in an essential way. In fact, starting with an infinite by
infinite matrix, it is generally difficult to turn it into a linear operator
in a Hilbert space unless it is assumed banded, see section \ref{OpTheory},
and the references cited there.

The purpose of this section is three-fold.

First to make precise the operator theory of banded infinite by infinite
matrices; and second to show that the infinite matrices used in representing
the operator algebra generated by Heisenberg's quantum mechanical momentum and
position observables consists of (infinite) banded matrices. Thirdly, we use
Heisenberg's (and Born's) computations to exhibit such banded operators which
are not essentially selfadjoint. The simplest such matrix $M$ is as follows:
let $P$ be Heisenberg's momentum operator and $Q$ the (dual) position
operator. Then we show that the monomial $M=QPQ$ is banded, but not
essentially selfadjoint. In fact, its deficiency indices are $(1,1)$.

\begin{definition}
\label{DefA1}Let $L$ be a countable \emph{(}typically infinite\emph{)}%
\thinspace set, and let $m:L\times L\rightarrow\mathbb{C}$ be a function on
$L\times L$. We say that $m$ is \emph{banded} iff for every $x\in L$, the set
\begin{equation}
\left\{  y\in L|m\left(  x,y\right)  \not =0\right\}  \label{EqA.1}%
\end{equation}
is finite.
\end{definition}

Let $\ell^{2}\left(  L\right)  $ be the sequence space with norm
\begin{equation}
\left\Vert v\right\Vert _{\ell^{2}}^{2}:\,=\sum_{x\in L}\left\vert v\left(
x\right)  \right\vert ^{2}<\infty\text{.} \label{EqA.2}%
\end{equation}

The sum on the right is the supremum of all the numbers $\sum_{x\in
F}\left\vert v\left(  x\right)  \right\vert ^{2}$ as $F$ ranges over all
finite subsets in $L$.

Let $\mathcal{D}$ be the dense subspace of all functions $v:L\rightarrow
\mathbb{C}$ such that the support set
\begin{equation}
\left\{  x\in L|v\left(  x\right)  \not =0\right\}  \label{EqA.3}%
\end{equation}
is finite. Equivalently, setting
\begin{equation}
\delta_{x}\left(  y\right)  =\left\{
\begin{array}
[c]{ll}%
1 & y=x\\
0 & y\not =x\text{;}%
\end{array}
\right.  \label{EqA.4}%
\end{equation}
the space $\mathcal{D}$ is then the linear span of the set of functions
$\{\delta_{x}|x\in L\}$; and these functions form an \textit{orthonormal}
basis for $\ell^{2}\left(  L\right)  $. Moreover, every Hilbert space
$\mathcal{H}$ is isomorphic to $\ell^{2}\left(  L\right)  $ for some set $L$.
The set $L$ is countable if and only if $\mathcal{H}$ is separable.

\begin{lemma}
\label{LemA.2}Let $m:L\times L\rightarrow\mathbb{C}$ be a banded function. For
$v\in\mathcal{D}\subseteq\ell^{2}\left(  L\right)  $, set
\begin{equation}
\left(  Mv\right)  \left(  x\right)  =\sum_{y\in L}m\left(  x,y\right)
v\left(  y\right)  \text{.} \label{EqA.5}%
\end{equation}

Then $M$ defines a linear operator $M:\mathcal{D}\rightarrow\mathcal{D}$, with
a well defined adjoint operator $M^{\ast}$. Moreover,
\begin{equation}
\mathcal{D}\subseteq\operatorname*{dom}\left(  M^{\ast}\right)  \label{EqA.6}%
\end{equation}
where $\operatorname*{dom}\left(  M^{\ast}\right)  $ is the domain of
$M^{\ast}$.
\end{lemma}

\begin{proof}
When $x\in L$ is fixed, the sum in (\ref{EqA.5}) is finite because the set
(\ref{EqA.1}) is finite by assumption. Using finiteness of both sets
(\ref{EqA.1}) and (\ref{EqA.3}) we conclude that $Mv$ in (\ref{EqA.5}) is in
$\mathcal{D}$ if $v$ is. And so, in particular, $Mv\in\ell^{2}\left(
L\right)  $; see (\ref{EqA.2}) and (\ref{EqA.4}).

To establish the inclusion \textquotedblleft$\subseteq$\textquotedblright\ in
(\ref{EqA.6}), we must show that for every $v\in\mathcal{D}$, there is a
constant $K=K\left(  v\right)  $ such that the following estimate holds:
\begin{equation}
\left\vert \langle Mu,v\rangle_{\ell^{2}}\right\vert \leq K\left\Vert
u\right\Vert _{\ell^{2}},\text{ for }\forall u\in\mathcal{D}\text{.}
\label{EqA.7}%
\end{equation}

The expression on the left in (\ref{EqA.7}) is
\begin{equation}
\underset{x,y\in L}{\sum\sum}\overline{m\left(  x,y\right)  }\overline
{u\left(  y\right)  }v\left(  x\right)  \text{.} \label{EqA.8}%
\end{equation}
But the terms in this double-sum vanish outside a finite subset in $L\times L$
an account of assumptions (\ref{EqA.1}) and (\ref{EqA.3}).

The modulus-square of the sum in (\ref{EqA.8}) is estimated by Schwarz by:
\[
\sum_{y\in L}\left\vert u\left(  y\right)  \right\vert ^{2}\sum_{y\in
L}\left\vert \sum_{x}m\overline{\left(  x,y\right)  }v\left(  x\right)
\right\vert ^{2}%
\]
which yields the desired estimate (\ref{EqA.7}).
\end{proof}

\begin{corollary}
\label{CorA.3}Let $M$ be a linear operator in a Hilbert space $\mathcal{H}.$
Then $M$ has a \emph{banded} matrix representation if and only if there is an
orthonormal basis \emph{(}ONB\emph{)}\thinspace in $\mathcal{H}$,
$\{e_{x}|x\in L\}$ such that the linear space $\mathcal{D}$ spanned by
$(e_{x})_{x\in L}$ is mapped into itself by $M$.
\end{corollary}

\begin{corollary}
\label{CorA.4}In that case the matrix entries of $M$ are indexed by $L\times
L$ as follows:%
\begin{equation}
m\left(  x,y\right)  :\,=\langle e_{x},Me_{y}\rangle\text{.} \label{EqA.9}%
\end{equation}

\end{corollary}

\begin{proof}
Only the conclusion (\ref{EqA.9}) is not contained in the lemma. Now suppose
some operator $M$ in $\mathcal{H}$ satisfies the conditions, and let $\left(
e_{x}\right)  _{x\in L}$ be the associated ONB. Then $Me_{y}\in\mathcal{H}%
\simeq\ell^{2}\left(  L\right)  $, so $Me_{y}=\sum_{x\in L}\langle
e_{x},Me_{y}\rangle_{\mathcal{H}}\,e_{x}$, and
\begin{equation}
\left\Vert Me_{y}\right\Vert _{\mathcal{H}}^{2}=\sum_{x\in L}\left\vert
\langle e_{x},Me_{y}\rangle\right\vert ^{2} \label{EqA.10}%
\end{equation}
holds by Parseval's formula. The conclusion (\ref{EqA.9}) follows.
\end{proof}

\begin{corollary}
\label{CorA.5}Let $G=(G^{\left(  0\right)  },G^{\left(  1\right)  })$ and
\[
c:G^{\left(  1\right)  }\rightarrow\mathbb{R}_{+}%
\]
be a graph system satisfying the axioms in section \ref{Assumptions}. Let
$\{\delta_{x}|x\in G^{\left(  0\right)  }\}$ be the canonical ONB in $\ell
^{2}(G^{\left(  0\right)  })$. Then the graph Laplacian has a corresponding
banded matrix representation as follows:
\begin{equation}
\langle\delta_{x},\Delta_{c}\delta_{y}\rangle=\left\{
\begin{array}
[c]{l}%
-c\left(  xy\right)  \text{ if }y\not =x\text{ and }y\sim x\\
\mathcal{B}_{c}\left(  x\right)  \text{ if }y=x\\
0\text{ if }y\not \sim x\text{ and }y\not =x\text{.}%
\end{array}
\right.  \label{EqA.11}%
\end{equation}

\end{corollary}

\begin{proof}
Recall the function
\begin{equation}
\mathcal{B}_{c}\left(  x\right)  :\,=\sum_{y\sim x}c\left(  xy\right)
\label{EqA.12}%
\end{equation}
on the right-hand side in (\ref{EqA.11}).

Since, for $v\in\mathcal{D}$, we have
\begin{equation}
\left(  \Delta_{c}v\right)  \left(  x\right)  :\,=\sum_{y\sim x}c\left(
xy\right)  \left(  v\left(  x\right)  -v\left(  y\right)  \right)
\text{,}\label{EqA.13}%
\end{equation}
setting $v=\delta_{y}$, we get
\[
\left(  \Delta_{c}\delta_{y}\right)  \left(  x\right)  =\left\{
\begin{array}
[c]{l}%
\mathcal{B}_{c}\left(  x\right)  \text{ if }y=x\\
-c\left(  xy\right)  \text{ if }y\sim x\\
0\text{ if }y\not \sim x
\end{array}
\right.
\]
from which the desired formula (\ref{EqA.11}) follows.\medskip
\end{proof}

Heisenberg introduced $\infty\times\infty$ matrix representations for the
operators of momentum $P$ and position $Q$ in quantum mechanics.

In the simplest case of one degree of freedom, they are as follows:
\[
\frac{1}{2}\left(
\begin{array}
[c]{cccccccccccccc}%
0 & 1 & 0 & 0 & 0 & \cdots & \cdots & \cdots & 0 & 0 & \cdots &  &  & \\
1 & 0 & \sqrt{2} &  &  &  &  &  & \cdots & \cdots & \cdots &  &  & \\
0 & \sqrt{2} & 0 &  &  &  &  &  & \ldots & \cdots &  &  &  & \\
0 & 0 & \sqrt{3} &  &  &  &  &  &  &  &  &  &  & \\
0 & 0 & 0 &  & \sqrt{n-2} & 0 & 0 & \cdots &  &  &  &  &  & \\
\vdots & \vdots & \vdots & \cdots & 0 & \sqrt{n-1} & 0 & \cdots &  &  &  &  &
& \\
&  &  & \cdots & \sqrt{n-1} & 0 & \sqrt{n} & \cdots &  &  &  &  &  & \\
&  &  & \cdots & 0 & \sqrt{n} & 0 & \cdots &  &  &  &  &  & \\
\vdots &  &  & \cdots & 0 & 0 & \sqrt{n+1} & \cdots & \ddots &  &  &  &  & \\
0 &  &  &  &  &  &  & \ddots & \ddots & 0 & 0 &  &  & \\
0 & 0 & \cdots &  &  &  &  & \ddots & 0 & 0 & 0 & \ddots &  & \\
\vdots & \vdots &  &  &  &  &  &  &  &  & \ddots & \ddots &  &
\end{array}
\right)
\]
and%
\[
\frac{1}{2i}\left(
\begin{array}
[c]{ccccccccccc}%
0 & -1 & 0 & \cdots &  &  &  &  & \cdots & \cdots & \cdots\\
1 & 0 & -\sqrt{2} & \cdots &  &  &  &  &  &  & \\
0 & \sqrt{2} & 0 &  & -\sqrt{n-2} & 0 & 0 &  &  &  & \\
\vdots & \vdots & \vdots &  & 0 & -\sqrt{n-1} & 0 &  &  &  & \\
&  &  &  & \sqrt{n-1} & 0 & -\sqrt{n} &  &  &  & \\
&  &  &  & 0 & \sqrt{n} & 0 &  &  &  & \\
\cdots &  &  &  & 0 & 0 & \sqrt{n+1} &  &  &  & \\
\cdots &  &  &  &  &  &  & \ddots & \vdots & \vdots & \vdots\\
\cdots &  &  &  &  &  &  & \cdots & \cdots & \cdots & \cdots
\end{array}
\right)  .
\]

Set $\mathbb{N}_{0}:\,=\{0,1,2,\ldots\}=\mathbb{Z}_{+}\cup\{0\}$, and
$\mathcal{H}:\,=\ell^{2}\left(  \mathbb{N}_{0}\right)  $. Then the two
matrices $P$ and $Q$ are represented by the following second order difference
operators, having the same form as our graph Laplacians (\ref{EqA.13}).
\begin{equation}
\left(  Pv\right)  \left(  n\right)  =\frac{1}{2}\left(  \sqrt{n-1}~v\left(
n-1\right)  +\sqrt{n}~v\left(  n+1\right)  \right)  \text{;} \label{EqA.14}%
\end{equation}
and
\begin{equation}
\left(  Qv\right)  \left(  n\right)  =\frac{1}{2i}\left(  \sqrt{n-1}~v\left(
n-1\right)  -\sqrt{n}~v\left(  n+1\right)  \right)  \text{,} \label{EqA.15}%
\end{equation}
for $\forall v\in\mathcal{D}$, $\forall n\in\mathbb{N}_{0}$; where
$i=\sqrt{-1}$.

It is well known that both $P$ and $Q$, as in (\ref{EqA.14}) and
(\ref{EqA.15}), are essentially selfadjoint.

It follows by the above lemma that
\begin{equation}
M:\,=QPQ \label{EqA.16}%
\end{equation}
is also a banded operator., referring to the canonical ONB $\left\{
e_{n}|n\in\mathbb{N}_{0}\right\}  $ in $\ell^{2}\left(  \mathbb{N}_{0}\right)
$.

Caution: All the operators $P,Q,$ and $M$ are unbounded, but densely defined;
see \cite{Jor77}, \cite{Jor78}, \cite{Sto51}.

\begin{proposition}
\label{PropA.6}The operator $M$ in \emph{(}\ref{EqA.16}\emph{)}\thinspace is
Hermitian, and has deficiency indices $(1,1)$; in particular is \emph{not}
essentially selfadjoint. In fact, it has many selfadjoint extensions; a
one-parameter family indexed by $\mathbb{T}$.
\end{proposition}

\begin{proof}
By the Stone-von Neumann uniqueness theorem, the two operators $P$ and $Q$ in
(\ref{EqA.14}) and (\ref{EqA.15}) are unitarily equivalent to the following
pair in the Hilbert space $L^{2}\left(  \mathbb{R}\right)  $ of all
square-integrable functions on the red line:
\begin{equation}
\left(  Pf\right)  \left(  x\right)  =\frac{1}{i}\frac{d}{dx}f\left(
x\right)  \text{,} \label{EqA.17}%
\end{equation}
and
\begin{equation}
\left(  Qf\right)  \left(  x\right)  =xf\left(  x\right)  \text{, for }\forall
f\in L^{2}\left(  \mathbb{R}\right)  ,~x\in\mathbb{R}\text{.} \label{EqA.18}%
\end{equation}
For domain $\mathcal{D}$ in (\ref{EqA.17}) and (\ref{EqA.18}), we may take
$\mathcal{D}:\,=C_{c}^{\infty}\left(  \mathbb{R}\right)  $, or the span of the
Hermite functions.

From the representations (\ref{EqA.14})-(\ref{EqA.15}), it follows that the
operator $M:\,=QPQ$ in (\ref{EqA.16}) commutes with a conjugation in the
Hilbert space; and so by von Neumann's theorem (see Remark \ref{Rem4.2}), it
has deficiency indices $\left(  n,n\right)  $. We will show that $n=1$. Hence
we must show that each of the equations $M^{\ast}v_{\pm}=\pm i\,v_{\pm}$ has a
one-dimensional solution space in $\mathcal{H}$.

Taking advantage of Schr\"{o}dinger's representation (\ref{EqA.17}%
)-(\ref{EqA.18}), we arrive at the corresponding pair of ODEs in $L^{2}\left(
\mathbb{R}\right)  $:
\begin{equation}
x\frac{d}{dx}\left(  xf\right)  =\pm f\left(  x\right)  \text{.}\label{EqA.19}%
\end{equation}
By symmetry, we need only to treat the first one.

A direct integration shows that
\begin{equation}
f\left(  x\right)  =\left\{
\begin{array}
[c]{ll}%
\frac{\exp\left(  \frac{-1}{x}\right)  }{x} & \text{if }x>0\\
0 & \text{if }x\leq0
\end{array}
\right.  \label{EqA.20}%
\end{equation}
solves (\ref{EqA.19}) in the case of \textquotedblleft+\textquotedblright\ on
the right hand side. Also note that (\ref{EqA.20}) is meaningful as all the
derivatives of $x^{-1}\exp(-\frac{1}{x})$ for $x\in\mathbb{R}_{+}$ tend to $0$
when $x\rightarrow0_{+}$. This means that the two separate expressions on the
right-hand side in (\ref{EqA.20}) \textquotedblleft patch\textquotedblright%
\ together differently at $x=0$.

By the reasoning alone, we conclude that $M$ has indices $\left(  1,1\right)
$. As a result of von Neumann's extension theory, the distinct selfadjoint
extensions of $M$ are then indexed by $\mathbb{T}=\{z\in\mathbb{C}|\left\vert
z\right\vert =1\}$. If $z\in\mathbb{T}$, and if $f_{\pm}$ are normalized
solutions to (\ref{EqA.19}), then the extension $M_{z}$ is determined by
\[
M_{z}\left(  f_{+}+zf_{-}\right)  =i\left(  f_{+}-zf_{-}\right)  \text{.}%
\]

\hfill
\end{proof}

\begin{example}
\label{ExampA.7}Let $P$ and $Q$ be the canonical momentum and position
operators; see \emph{(}\ref{EqA.14}\emph{)}-\emph{(}\ref{EqA.15}\emph{)}, and
let
\begin{equation}
H:\,=P^{2}-Q^{4} \label{EqA.21}%
\end{equation}
be the Hamiltonian of a \textquotedblleft particle-wave\textquotedblright\ in
one degree of freedom, corresponding to a repulsive $x^{4}$ potential. Then
the reasoning from above shows that $H$ is a banded $\infty\times\infty$
matrix. As an \emph{operator} in $\ell^{2}\left(  \mathbb{Z}\right)  $, $H$
has deficiency indices $\left(  2,2\right)  $.
\end{example}

\bibliographystyle{alpha}
\bibliography{Jorgen}

\end{document}